\documentclass[a4paper]{easychair}
\usepackage[T1]{fontenc}
\usepackage{listings}
\usepackage{graphicx}
\usepackage{amsmath}
\usepackage{bussproofs}
\usepackage{mathtools}

\usepackage{amssymb}
\usepackage{xspace}

\newtheorem{definition}{Definition}
\newtheorem{therm}{Theorem}

\newtheorem{example}{Example}
\newtheorem{lemma}{Lemma}
\newtheorem{corollary}{Corollary}

\usepackage{ifthen}

\newcommand{\inttype}{\textit{Int}}
\newcommand{\floattype}{\textit{Float}}
\newcommand{\booltype}{\textit{Bool}}
\newcommand{\strtype}{\textit{String}}
\newcommand{\btype}{b}           
\newcommand{\vtype}{\alpha}      
\newcommand{\ttype}{\sigma}      
\newcommand{\ltype}{\rho}        
\newcommand{\mtype}{\tau}        
\newcommand{\etype}{\gamma}      
\newcommand{\gkind}{\kappa}      
\newcommand{\ukind}{\mathcal{U}} 
\newcommand{\glab}{l}            

\newcommand{\gterm}{M} 
\newcommand{\vterm}{x} 
\newcommand{\cterm}{k} 

\newcommand{\val}{V} 
\newcommand{\vvalue}{\textit{value}} 
\newcommand{\pred}{\textit{p}} 

\newcommand{\eval}[1]{\mathcal{E}[#1]}
\newcommand{\evall}{\mathcal{E}}
\newcommand{\hole}{\ }

\newcommand{\geqs}{E}             
\newcommand{\ksubs}[2]{(#1,#2)}   
\newcommand{\kuni}[2]{(#1,#2)}    
\newcommand{\tenv}{\Gamma}        
\newcommand{\kenv}{K}             
\newcommand{\stvar}{\mathbb{V}}   
\newcommand{\stbase}{\mathbb{B}}  
\newcommand{\slab}{\mathcal{L}}   

\newcommand{\letin}[3]{\text{let} \ #1 = #2 \ \text{in} \ #3} 
\newcommand{\modif}[3]{\text{modify}(#1,#2,#3)} 

\newcommand{\subs}{S}
\newcommand{\id}{\textit{id}}
\newcommand{\dom}[1]{\textit{dom}(#1)}
\newcommand{\ftv}[1]{\textit{FTV}(#1)}
\newcommand{\eftv}[2]{\textit{EFTV}(#1,#2)}
\newcommand{\cls}[3]{\textit{Cls}(#1,#2,#3)}
\newcommand{\unify}[2]{\textit{U}(#1,#2)}     
\newcommand{\infer}[3]{\textit{WK}(#1,#2,#3)} 

\newcommand{\spec}{\mathsf{spec}}
\newcommand{\gev}{\textit{ge}}
\newcommand{\ev}{\textit{e}}

\newcommand{\itype}{\inttype}
\newcommand{\bltype}{\booltype}
\newcommand{\atype}[2]{#1 \rightarrow #2} 
\newcommand{\ctype}[2]{(#1 \times #2)}    
\newcommand{\true}{\texttt{true}}
\newcommand{\false}{\texttt{false}}

\newcommand{\app}[2]{#1 \ #2}                                               
\newcommand{\abs}[2]{\lambda #1.#2}                                         
\newcommand{\letw}[3]{\text{let} \ #1 = #2 \ \text{in} \ #3}                
\newcommand{\letEv}[3]{\text{letEv} \ #1 = #2 \ \text{in} \ #3}             
\newcommand{\sel}[2]{#1.#2}                                                 
\newcommand{\modw}[3]{\text{modify}(#1,#2,#3)}                              
\newcommand{\cond}[3]{\text{if} \ #1 \ \text{then} \ #2 \ \text{else} \ #3} 

\newcommand{\proj}{\pi}

\newcommand{\minus}{\ \text{-} \ }
\newcommand{\greater}{\ \text{>} \ }
\newcommand{\lesser}{\ \text{<} \ }
\newcommand{\divi}{\ \text{/} \ }
\newcommand{\logicand}{\ \text{and} \ }

\newcommand{\greaterfloattype}{\atype{\floattype}{\atype{\floattype}{\booltype}}}
\newcommand{\lesserfloattype}{\atype{\floattype}{\atype{\floattype}{\booltype}}}
\newcommand{\logicandtype}{\atype{\booltype}{\atype{\booltype}{\booltype}}}

\newcommand{\danger}{\text{fire\_danger}}
\newcommand{\wind}{\text{wind}}
\newcommand{\humidity}{\text{humidity}}
\newcommand{\precipitation}{\text{precipitation}}
\newcommand{\temperature}{\text{temperature}}
\newcommand{\location}{\text{location}}
\newcommand{\fst}{\text{fst}}
\newcommand{\snd}{\text{snd}}

\newcommand{\porto}{\texttt{"Porto"}}

\newcommand{\high}{\texttt{"high"}}
\newcommand{\low}{\texttt{"low"}}

\newcommand{\WeatherInfo}{\text{WeatherInfo}}
\newcommand{\FireDanger}{\text{FireDanger}}

\newcommand{\kind}[1]{\{\!\{#1\}\!\}}

\newcommand{\ra}{\rightarrow}
\newcommand{\Ra}{\Rightarrow}

\newcommand{\EVL}{{\tt EVL}\xspace}

\newcommand{\shterm}{M}    
\newcommand{\shtype}{\tau} 
\newcommand{\shsubs}{S}    

\newcommand{\cmiguel}[2]{\miguel{#2}}
\newcommand{\miguel}[1]{\textcolor{red} {#1}}

\usepackage{ifthen}
\usepackage{bm}

\title{EVL: a typed functional language for event processing\thanks{This work is financed by National Funds through the Portuguese funding agency, FCT - Fundação para a Ciência e a Tecnologia, within project UIDB/50014/2020. We also acknowledge  support by the London Mathematical Society (SC7-1920-16).}}

\author{Sandra Alves\inst{1,3}
\and Maribel Fern\'andez\inst{2}
\and Miguel Ramos\inst{3}}
    
\institute{CRACS - INESCTEC \\ Porto, Portugal\\
\email{sandra@fc.up.pt}\\\ \\
\and
Dept.\ of Informatics\\ King's College London, London WC2B 4BG, U.K.\\ \email{maribel.fernandez@kcl.ac.uk}\\\ \\
\and 
DCC-FCUP \\ University of Porto, Porto, Portugal\\
   \email{jmiguelsramos@gmail.com}
 }

\begin{document}

\maketitle

\begin{abstract} 
  We define \EVL, a minimal higher-order functional language to deal with generic events. The notion of generic event extends the well-known notion of event traditionally used in a variety of areas, such as database management, concurrency, reactive systems and cybersecurity. Generic events were introduced in the context of a metamodel to specify obligations in access control systems. Event specifications are represented as records and we use polymorphic record types to type events in \EVL. We show how the higher-order capabilities of \EVL can be used in the context of Complex  Event  Processing (CEP), to define higher-order parameterised functions that deal with the usual CEP techniques.
\end{abstract}

\section{Introduction}\label{intro}

Complex Event Processing (CEP), or simply event processing, refers to a set of techniques used to deal with event streams, including event identification, classification and response. Events are occurrences  of actions, or happenings, and a variety of languages to process events have been developed over the years~\cite{BargaGAH07,Luckham02,Etzion10}.  
Event processing is a key component of Internet-of-Things applications, which need to identify events in the streams of data generated by sensors and react appropriately. 
In  critical domains (e.g., healthcare) it is important to be able to prove properties of applications (correctness, security, safety, etc.) and for this we need languages with a formal semantics. In this paper, we focus on the functional paradigm, for which advanced techniques have been developed to reason about programs: we develop a functional language to specify and process events, with a polymorphic type system inspired by Ohori's record calculus~\cite{Ohori95}.


In the context of security, and in particular when modelling access control, it is often the case that granting or denying access to certain resources depends on the occurrence of a particular event~\cite{bsw08,HiltyPBSW07,Bertino:2001}. This is even more crucial in access control systems dealing with obligations, where the status of a particular obligation is usually defined in terms of event occurrences in the system, and several models that deal with obligations have to deal in some way with the notion of event. The Category-Based metamodel for Access Control and Obligations (CBACO~\cite{AlvesDM14}), axiomatises the notion of obligation based on generic relations involving events and event intervals. A key distinction is made between event schemes, which provide a general description of the kind of events that can occur in a particular system, and specific events, which represent actual events that have occurred. 

Events can take various forms, depending on the system that is being considered (for example,  messages exchanged over a network, actions performed by users of the system, occurrences of physical phenomena such as a disk error or a fire alarm, etc). 
To deal with event classification in a uniform way, Alves et. al.~\cite{AlvesBF15} defined a general term-based language for events.  In this language, events are represented as typed-terms built from a user-defined signature, that is, a particular set of typed function symbols that are specific to the system modelled. With this approach it is possible to define general functions to implement event typing and to compute event intervals, without needing to know the exact type of events.
A compound event~\cite{AlvesBF15} links a set of events that  occur separately in the history, but should be identified as a single event occurrence. For simplicity, in~\cite{AlvesBF15} compound events were assumed to appear as a single event in history, leaving a more detailed and realistic treatment of compound events for future work. The notion of compound or composite event is also a key feature in CEP systems, which put great emphasis on the ability to detect complex patterns of incoming streams of events and establish sequencing and ordering relations.

Types  were used in~\cite{AlvesBF15} not only to ensure that terms representing events respect the type signature specific to the system under study, but also to formally define the notion of event instantiation, associating specific events to generic events through an implicit notion of subtyping, inspired by Ohori's system of polymorphic record types~\cite{Ohori95}. Because of the  implicit subtyping rule for typing records, the system defined in~\cite{AlvesBF15} allowed for type-checking of event-specification, but not for dealing with most general types for event specifications. 

To facilitate the specification and processing of events, including compound events, in~\cite{AlvesFR20} we introduced  \EVL, a higher-order polymorphic type system, which   is both a restriction and an extension of Ohori's  polymorphic record calculus~\cite{Ohori95}.  Although \EVL is very much based on that system, it is not meant to be a general system, but rather a language purposely designed for dealing with events. In this paper we complete the definition of the language by providing a call-by-value operational semantics for \EVL using evaluation contexts, which takes into account the domain-specific aspects of the language. 

Languages traditionally used in event processing systems are usually derived from relational languages, in particular, relational algebra and SQL, extended with additional ad-hoc operators to better support information flow, or imperative programming. This paper exploits the potential of the functional paradigm in this context, both at the level of the type-system, as well as in the higher-order features of the language.
The main contributions of this paper are:
\begin{itemize}
\item The design of \EVL: a minimal higher-order functional language with polymorphic record types, tailored to the specific tasks of event processing. 
\item A sound and complete type inference algorithm for \EVL, defined as both an extension and a restriction of the ML-style record calculus in~\cite{Ohori95}.

\item A Call-by-Value operational semantics, using evaluation contexts, for which we show the type preservation property. 

\item A comprehensive study of the \EVL higher-order/functional capabilities and its application in the context of CEP.
\end{itemize} 
This paper is a revised and extended version of \cite{AlvesFR20}, where the term language and the type system  were presented. Here we provide a formal operational semantics for the language, prove that evaluation preserves types, and illustrate with examples the expressive power of the language.

\paragraph*{Overview} In Section~\ref{sec:background} we summarise the main notions of event processing and type systems needed to make the paper self contained. In Section~\ref{sec:evl} we define the \EVL language and its set of types. In Section~\ref{sec:ta} we define a type system for \EVL and in Section~\ref{sec:cbv} we present the call-by-value operational semantics and prove type-preservation. In Section~\ref{sec:ti} we present a type inference algorithm, which is proved to be sound and complete. In Section~\ref{sec:cep} we explore \EVL's capabilities. We discuss related work in Section~\ref{sec:rw} and we finally conclude and discuss further work in Section~\ref{sec:conc}.


\section{Background}
\label{sec:background}
In this section we recall basic notions on type systems and event processing that will be used in the design of \EVL. We refer  to~\cite{Ohori95} for more details on record typing and to~\cite{Luckham02} for more details on event processing.

 \subsection{Polymorphic records}
A convenient way to construct data types is using records, which can be seen as tuples where the individual components are labeled. Record calculi are generally used to model programming features such as objects and module systems. One such record calculus is the one described by Ohori in~\cite{Ohori95}, based on a kinded quantification, which extends the standard type system for parametric polymorphism by Damas and Milner~\cite{DamasM82}, with primitives for record creation, field selection and field update. A record term is a term of the form $\{\glab_1 = M_1, \ldots, \glab_n = M_n\}$ that represents a structure with $n$ fields $\glab_1, \ldots, \glab_n$ and values $M_1, \ldots, M_n$, respectively.  The term  $\sel{M}{\glab}$ represents the selection of field $\glab$ from the structure $M$ and the term $\modif{M_1}{\glab}{M_2}$ represents changing the value $M_1$ of field $\glab$ to $M_2$.

\begin{example}
    \label{ex:selmodif}
    The two following terms are valid terms in Ohori's record calculus:
    \begin{align}
        & \abs{x y}{\letin{\textit{getName}}{\abs{z}{(\sel{z}{\textit{name}})}}{\app{\textit{getName}}{\{\textit{name} = x, \textit{address} = y\}}}} \label{eq:sel} \\
        & \abs{x y z}{\letin{\textit{update}}{\abs{x y}{\modif{x}{\textit{address}}{y}}}{\app{(\app{\textit{update}}{\{\textit{name} = x, \textit{address} = y\}})}{z}}} \label{eq:modif}
    \end{align}
\end{example}

Note that in the term~\eqref{eq:sel} we apply the function $\textit{getName}$ to the concrete structure $\{\textit{name} = x, \textit{address} = y\}$ and in the term~\eqref{eq:modif}  we apply function $update$ again to the concrete structure $\{\textit{name} = x, \textit{address} = y\}$. That being said, it is easy to see that function $\textit{getName}$ can be applied to any structure that contains the field $\textit{name}$ and that function $\textit{update}$ can be applied to any structure that contains the field $\textit{address}$. There are various ways to deal with this type of record polymorphism (e.g. qualified types, row variables...). This is achieved in~\cite{Ohori95} through the use of a system of \emph{kinds}. Quantified types (or type schemes) are of the form $\forall\alpha::\gkind.\sigma$,  where the type variable $\vtype$ is restricted by a kind $\gkind$.  A kind represents a set of types and  can be either the universal kind $\ukind$, representing all possible types, or a kind of the form $\kind{\glab_1 : \tau_1, \ldots, \glab_n : \tau_n}$ representing the types of records that have, at least, fields $\glab_1, \ldots, \glab_n$ of types $\tau_1, \ldots, \tau_n$, respectively.

 \subsection{Events}
In event processing applications, many events  have a similar structure and a similar meaning. Consider a temperature sensor: all of the events produced by it have the same kind of information, such as temperature reading, timestamp and maybe location, but with possibly different values. This relationship was formally defined in~\cite{AlvesBF15} as that between \emph{Generic} and \emph{Specific} events. 
 We now give some key notions on events that will be central to the definition of our language. We consider events as particular actions or happenings occurring at a particular time.
\begin{definition}[Event Specifications] 
\label{def:evspec}
Given a set of terms $M_1,\dots,M_n$, defined in a particular language, an \emph{event specification}, denoted $\spec$, is a term of the form $\{l_1= M_1,\dots,  l_n=M_n\},\  n > 0$, 
representing a structure with labels $l_1,\dots, l_n$ and values $M_1,\dots,M_n$ respectively. 
An event specification without term variable occurrences is called a ground event specification.
\end{definition}
In our language, event specifications  will be typed using record types, following  Ohori~\cite{Ohori95}. We distinguish between \emph{events} and \emph{generic events} (or event schemes), where the former correspond to specific happenings or occurrences and the latter represent sets of events that can occur in a particular system. 

\begin{definition}[Event] A \emph{(specific) event} is a ground event specification that represents a particular action/happening, occurring in a system. 
\end{definition}
\begin{definition}[Generic Event] A \emph{generic event} (or \emph{event scheme}) represents a set of events, defined as
$\gev[x_1,\dots,x_n] = \spec$, where $x_1,\dots,x_n$ are the variables occurring in $\spec$.
\end{definition}
Specific events $\ev$ are associated to generic events $\gev$, by an instantiation relation $\vdash_{\theta} \ev :: \gev$.  The instantiation relation can be syntactic ($\ev$ is obtained from $\gev$ by replacing the variables in $\gev$ by terms through a substitution mapping $\theta$), but can also be a semantic instantiation that may require some computation.

 \subsection{Complex Event Processing }
   The area of CEP comprises a series of techniques to deal with streams of events such as event processing, detection of patterns and relationships, filtering, transformation and abstraction, amongst others.  See~\cite{Etzion10} for a detailed reference on the area. 

Event processing agents are classified according to the actions that they perform to process incoming events. We are now going to look into the different types of event processing agents in a little more depth. All of the definitions that we are going to present can be found in~\cite{Etzion10}.

\begin{definition}[Filter event processing agent]
  A \emph{filter agent} is an event processing agent that performs filtering only, so it does not transform the input event.
\end{definition}
\begin{definition}[Transformation event processing agent]
  A \emph{transformation agent} is an event processing agent that includes a derivation step, and optionally also a filtering step.
\end{definition}
\begin{definition}[Translate event processing agent]
  A \emph{translate agent} can be used to convert events from one type to another, or to add, remove, or modify the values of an event’s attributes.
\end{definition}
\begin{definition}[Aggregate event processing agent]
 An \emph{aggregate agent} takes a stream of incoming events and produces an output event that is a map of the incoming events.
\end{definition}
\begin{definition}[Compose event processing agent]
  A \emph{compose agent} takes two streams of incoming events and processes them to produce a single output stream of events.
\end{definition}
\begin{definition}[Pattern Detect event processing agent]
  A \emph{pattern detect agent} performs a pattern matching function on one or more input streams. It emits one or more derived events if it detects an occurrence of the specified pattern in the input events.
\end{definition}
The notions of specific and generic events are also a key aspect in CEP, where instead of defining the structure of each event individually, one wants to be able  specify the structure of an entire class of events. Generic events can then be related to other (generic or specific) events through semantic relations. In~\cite{Etzion10}, these relationships where classified into four types: membership, generalization, specialization and retraction. We will show latter how these relationships can be dealt with in \EVL.
\section{The \EVL typed language}
\label{sec:evl}
In this section we introduce \EVL, a minimalistic typed language to specify events. \EVL is an  extension of the $\lambda$-calculus that includes records, a flexible data structure that is used here to deal with event specifications (Definition~\ref{def:evspec}). 
We assume some familiarity with the $\lambda$-calculus (see~\cite{Barendregt85} for a detailed reference).


We start by formally defining the set of \EVL terms. In the following,  $x,y,z,\dots$ range over a countable set of variables and $\glab,\glab_1,\dots$ range over a countable set  $\slab$ of labels. 

\begin{definition}The set of \EVL terms is given by the following grammar:
\[
\begin{array}{lcl}
    \gterm & ::= & \cterm^\btype \mid \vterm \mid \app{\gterm}{\gterm} \mid \abs{\vterm}{\gterm} \mid \cond{\gterm}{\gterm}{\gterm} \\
    & & \letw{\vterm}{\gterm}{\gterm} \mid \letEv{\vterm}{\gterm}{\gterm} \\
    & & \{\glab =\gterm, \dots, \glab=\gterm\} \mid \gterm.\glab \mid \modw{\gterm}{\glab}{\gterm}
\end{array}
\]
where $\cterm$ is a constant and $\btype$ is a constant type from the set $\stbase$ of constant types. We will assume the existence of two constants $\true$ and $\false$ of type $\booltype$.
\end{definition}

%
%

\textbf{Notation:} We will use the notation $let \ x \ \vterm_1 \dots \vterm_n = M$ and $letEv \ x \ \vterm_1 \dots \vterm_n = M$ for $let \ x = \lambda\vterm_1\dots\lambda\vterm_n.\gterm$ and $letEv \ x = \lambda\vterm_1\dots\lambda\vterm_n.\gterm$, respectively.
 As an abuse of notation, in examples, we will use more meaningful names for functions, labels and events. Furthermore, event names will always start with a capital letter, to help distinguish them from functions.

We choose not to add other potentially useful constructors to the language, for instance pairs and projections, since we are aiming at a minimal language. Nevertheless, we can easily encode pairs $(\gterm_1, \gterm_2)$ and projections $\proj_1 \gterm$ and $\proj_2 \gterm$ in our language by means of records of the form $\{\fst=\gterm_1, \snd=\gterm_2\}$ and $\gterm.\fst$, $\gterm.\snd$, respectively. This can trivially be extended to tuples in general, and we will often use this notation when writing examples.
\begin{example}In this simple example, \emph{FireDanger} reports the fire danger level of a particular location.

\begin{lstlisting}[mathescape=true, basicstyle=\small, upquote=true]
    letEv FireDanger = $\lambda$l$\lambda$d.{location = l, fire_danger = d} in
    FireDanger $\porto^\strtype$ $\low^\strtype$
\end{lstlisting}
\end{example}
To make our examples more readable, we will also use the following terms, abbreviating list construction:
\begin{lstlisting}[mathescape=true, basicstyle=\small]
    nil = {empty = $\true^\booltype$}
\end{lstlisting}
\begin{lstlisting}[mathescape=true, basicstyle=\small]
    cons x list = {empty = $\false^\booltype$, head = x, tail = list}.
\end{lstlisting}
Note that, much like what happens with tuples, the type of a particular list in this notation will be closely related to the size of the list in question. A more realistic approach is to add lists and list-types as primitive notions in the language, but, as we mentioned before, we are focusing on a minimal language.
Furthermore, we will often use constants (numbers, booleans, strings, etc) and operators (arithmetic, boolean, etc) in our examples. However, following the minimalistic approach, we do not add constants/operators to the grammar and instead use free variables to represent them. Again, in a more general approach we could extend the grammar with other data structures and operators, for numbers, booleans, lists, etc.

%
\begin{example}\label{ex:fire_check}Consider the following example illustrating the definition of a generic event \emph{FireDanger} and of a function \emph{check} that determines if there is the danger of a fire erupting in a particular location, using the weather information associated with that location. Function \emph{check}  creates an appropriate instance of  \emph{FireDanger} to report the appropriate fire danger level.
    \begin{lstlisting}[mathescape=true, basicstyle=\small]
    letEv FireDanger l d = {location = l, fire_danger = d} in
    let check x = if (x.temperature > $29.0^\floattype$ and x.wind > $32.0^\floattype$ 
                      and x.humidity < $20.0^\floattype$ and x.precipitation < $50.0^\floattype$) 
                  then FireDanger x.location $\high^\strtype$
                  else FireDanger x.location $\low^\strtype$ in
        check {temperature = $10.0^\floattype$, wind = $20.0^\floattype$, humidity = $30.0^\floattype$, 
               precipitation = $10.0^\floattype$, location = $\porto^\strtype$} 
    \end{lstlisting}
\end{example}

We now define the set of types for the \EVL language. We use record types to type labelled structures.
We assume a finite set $\stbase$ of constant types  and a countable set  $\stvar$ of type variables, and we will use $\btype ,\btype_1,\dots$, $\vtype, \vtype_1, \dots$ and $\gkind,\gkind_1,\dots$ to denote  constant types, type variables and kinds, respectively.
The set $\stbase$ of constant types will always contain the type $\bltype$.
\begin{definition} The sets of types $\ttype$ and kinds $\gkind$ are specified by the following grammar. 
  \begin{align*}
    \ttype &::= \mtype \mid \forall \vtype::\gkind.\ttype \\
    \mtype &::= \vtype \mid \btype \mid \atype{\mtype}{\mtype} \mid \{\glab:\mtype, \dots, \glab:\mtype\} \\
    \ltype &::= \vtype \mid \btype \mid \atype{\mtype}{\ltype} \\
    \etype &::= \atype{\mtype}{\{\glab:\ltype, \dots, \glab:\ltype\}} \mid \{\glab:\ltype, \dots, \glab:\ltype\}\\
    \gkind & ::= \ukind \mid \kind{\glab:\mtype, \dots, \glab:\mtype}
  \end{align*}
\end{definition}
  Following Damas and Milner’s type system, we divide the set of types into \emph{monotypes} (ranged over by $\tau$) and \emph{polytypes} (of the form $\forall \vtype::\gkind.\ttype$). More precisely,   $\ttype$ represents all types and $\mtype$ represents all monotypes. We denote by $\ltype$ (included in $\mtype$) the type of  event fields, and by $\etype$ (also included in $\mtype$) the type of  event definitions. 
  This distinction is necessary to adequately type event definitions and its purpose will become clear in the definition of the typing system.

We do not allow nested events, and to that end we clearly separate types for event definitions, denoted by $\gamma$, and which are a subset of the general types denoted by $\tau$. However, we do allow for nested records of general type. The following is an example of a term that is typed with a nested record type:
\begin{lstlisting}[mathescape=true, basicstyle=\small]
    {empty = $\false^\booltype$, head = $1^\inttype$, 
     tail = {empty = $\false^\booltype$, head = $2^\inttype$, tail = {empty = $\true^\booltype$}}}.
\end{lstlisting} 

Let $F$ range over functions from a finite set of labels to types. We write $\{F\}$ and $\kind{F}$ to denote the record type identified by $F$ and the record kind identified by $F$, respectively. For two functions $F_1$ and $F_2$ we write $F_1 \pm F_2$ for the function $F$ such that $\dom{F} = \dom{F_1} \cup \dom{F_2}$ and such that for $l \in \dom{F}$, $F(l) = F_1(l)$ if $l \in \dom{F_1}$; otherwise $F(l) = F_2(l)$.

\textbf{Notation:} Following the notation for pairs introduced above, we write  $(\ttype_1 \times \ttype_2)$ for the product type corresponding to  $\{\fst : \ttype_1, \snd : \ttype_2\}$.

A typing environment $\tenv$ is a set of statements $\vterm : \ttype$ where all subjects $x$ are distinct. We write $\dom{\tenv}$ to denote the domain of a typing environment $\tenv = \{\vterm_1 : \ttype_1, \dots, \vterm_n : \ttype_n\}$, which is the set  $\{\vterm_1, \dots, \vterm_n\}$. The type of a variable $\vterm_i \in \dom{\tenv}$ is $\tenv(\vterm_i) = \ttype_i$, and we write $\tenv_x$ to denote $\tenv\setminus \{x:\tenv(x)\}$. A kinding environment $\kenv$ is a set of statements $\vtype::\gkind$. Similarly, the domain of a kinding environment $\kenv = \{\vtype_1::\gkind_1, \dots, \vtype_n::\gkind_n\}$, denoted $\dom{K}$, is the set $\{\vtype_1, \dots, \vtype_n\}$ and the kind of a type $\vtype_i \in \dom{\kenv}$ is $\kenv(\vtype_i) = \gkind_i$.
A type variable $\alpha$ occurring in a type/kind is bound, if it occurs under the scope of a $\forall$-quantifier on $\alpha$, otherwise it is free. We denote by $\ftv{\ttype}$ ($\ftv{\gkind}$) the set of free variables of $\ttype$ (respectively, $\gkind$). We say that a type $\ttype$ and a kind $\gkind$ are well-formed under a kinding environment $\kenv$ if $\ftv{\ttype} \subseteq \dom{K}$ and $\ftv{\gkind} \subseteq \dom{K}$, respectively. A typing environment $\tenv$ is well-formed under a kinding environment $\kenv$, if $\forall \vterm \in \dom{\tenv}$, $\tenv(\vterm)$ is well-formed under $\kenv$. A kinding environment $\kenv$ is well-formed, if $\forall \vtype \in \dom{K}, \ftv{\kenv(\vtype)} \subseteq \dom{K}$. This reflects the fact that every free type variable in an expression has to be restricted by a kind in the kinding environment. Therefore, every type variable is either restricted by the kind in the type scheme or by a kind in the kinding environment.

Furthermore, we consider the set of essentially free type variables of a type $\ttype$ under a kinding environment $\kenv$ (denoted as $\eftv{\kenv}{\ttype}$) as the smallest set such that, $\ftv{\ttype} \subseteq \eftv{\kenv}{\ttype}$ and if $\vtype \in \eftv{K}{\ttype}$, then $\ftv{K(\vtype)} \subseteq \eftv{K}{\ttype}$. This reflects the fact that a type variable $\vtype$ is essentially free in $\ttype$ under a kinding environment $\kenv$, if $\vtype$ is free in $\ttype$ or in a restriction in $\kenv$.

\begin{definition}
Let $\mtype$ be a monotype, $\gkind$ a kind, and $\kenv$ a kinding environment. Then we say that $\mtype$ has kind $\gkind$ under $\kenv$ (written $\kenv \Vdash \mtype::\gkind$), if $\mtype::\gkind$ can be obtain by applying the following rules:
\begin{align*}
  \kenv &\Vdash \mtype::\ukind \ \text{for all} \ \mtype \ \text{well-formed under} \ \kenv \\
  \kenv &\Vdash \vtype::\kind{l_1:\mtype_1, \dots, l_n:\mtype_n} \ \text{if} \ \kenv(\vtype) = \kind{\glab_1 : \mtype_1, \dots, \glab_n:\mtype_n,\dots} \\
  \kenv &\Vdash \{\glab_1:\mtype_1, \dots, \glab_n:\mtype_n, \dots\}::\kind{\glab_1:\mtype_1, \dots, \glab_n:\mtype_n} \\
        &\quad{} \text{if} \ \{\glab_1:\mtype_1, \dots, \glab_n:\mtype_n, \dots\} \ \text{is well-formed under} \ \kenv
\end{align*}
Note that, if $\kenv \Vdash \ttype :: \gkind$, then $\ttype$ and $\gkind$ are well-formed under $\kenv$.
\end{definition}
\begin{example}
Let $\mtype = \atype{\vtype_1}{\{\glab_2:\itype, \glab_3:\ctype{\vtype_2}{\vtype_3}\}}$. Then, $\mtype$ is well-formed under $\kenv_1 = \{\vtype_1::\ukind, \vtype_2::\ukind, \vtype_3::\ukind\}$, because $\ftv{\mtype} \subseteq \dom{\kenv_1}$, but not under $\kenv_2 = \{\vtype_1::\ukind, \vtype_2::\ukind\}$, because $\vtype_3 \not \in \dom{\kenv_2}$, and, therefore, $\ftv{\mtype} \not \subseteq \dom{\kenv_2}$. Because $\mtype$ is well-formed under $\kenv_1$, we can write  $\kenv_1 \Vdash \mtype :: \ukind$.
\end{example}

\section{Type assignment}
\label{sec:ta}
We now define how types are assigned to \EVL terms. Because we are dealing with polymorphic type schemes, we need to define the notion of \emph{generic instance} for which we first need to discuss well-formed substitutions. 

A substitution $\subs = [\ttype_1/\vtype_1, \dots, \ttype_n/\vtype_n]$ is  \emph{well-formed under a kinding environment} $\kenv$, if for all $\vtype \in \dom{\subs}$, $\subs(\vtype)$ is well-formed under $\kenv$. This reflects the fact that applying a substitution to a type that is well-formed under a kinding environment $K$, should result in a type that is also well-formed under $K$. A \emph{kinded substitution} is a pair $(\kenv, \subs)$ of a kind assignment $\kenv$ and a substitution $S$ that is well-formed under $\kenv$. This reflects the fact that a substitution $\subs$ should only be applied to a type that is well-formed under $\subs$, such that the resulting type is kinded by $\kenv$.
\begin{example}
Let $\subs = [\alpha_2/\alpha_1]$ be a substitution. Then $\dom{\subs} = \{\alpha_1\}$, $\subs(\alpha_1) = \alpha_2$, and $\ftv{\alpha_2} = \{\alpha_2\}$. For the kinding environment $\kenv_1 = \{\alpha_2 :: \gkind\}$, we have that $\subs$ is well-formed under $\kenv_1$, since $\alpha_2 \in \dom{\kenv_1}$. On the other hand, for the kinding environment $\kenv_2 = \{\alpha_3 :: \gkind\}$, we have that $\subs$ is not well-formed under $\kenv_2$, since $\alpha_2 \not\in \dom{\kenv_2}$.
\end{example}
\begin{definition}
  We say that a kinded substitution $(\kenv_1, \subs)$ respects a kinding environment $\kenv_2$, if $\forall \alpha \in \dom{\kenv_2}, \kenv_1 \Vdash \subs(\alpha) :: \subs(\kenv_2(\alpha))$.
\end{definition}

\begin{example}
Let $\kenv_1 = \{\vtype_1 :: \kind{\glab_1 : \vtype_2}, \vtype_2::\ukind\}$ and $S = [\{\glab_1 : \itype\}/\vtype_1]$. Then, the restricted substitution $(\kenv_1, \subs)$ respects $\kenv_2 = \{\vtype_1 :: \kind{\glab_1 : \itype}\}$,
because for $\dom{\kenv_2} = \{\vtype_1\}$, we have:
\begin{align*}
  \kenv_1 &\Vdash \subs(\vtype_1) :: \subs(\kenv_2(\vtype_1)) \\
  \kenv_1 &\Vdash \subs(\vtype_1) :: \subs(\kind{\glab_1 : \itype}) \\
   \kenv_1 &\Vdash \subs(\vtype_1) :: \kind{\glab_1 : \subs(\itype)} \\
    \kenv_1 &\Vdash \{\glab_1 : \itype\} :: \kind{\glab_1 : \itype}
\end{align*}
\end{example}
\begin{lemma}
  \label{lem0}
  If $\ftv{\ttype} \subseteq \dom{\kenv}$ and $(\kenv_1, \subs)$ respects $\kenv$, then $\ftv{\subs(\ttype)} \subseteq \dom{\kenv_1}$.
\end{lemma}

\begin{lemma}
    \label{lem1}
    If $\kenv \vdash \ttype :: \gkind$, and a kinded substitution $(\kenv_1,\subs)$ respects $\kenv$, then $\kenv_1 \Vdash \subs(\ttype) :: \subs(\gkind)$.
\end{lemma}

\begin{proof}
  Since $\kenv \Vdash \ttype :: \gkind$, then it is of the form:
\begin{align*}
    \kenv &\Vdash \mtype::\ukind \ \text{for all} \ \mtype \ \text{well formed under} \ \kenv \\
    \kenv &\Vdash \vtype::\kind{l_1:\mtype_1, \dots, l_n:\mtype_n} \ \text{if} \ \kenv(\vtype) = \kind{\glab_1 : \mtype_1, \dots, \glab_n:\mtype_n, \dots} \\
    \kenv &\Vdash \{\glab_1:\mtype_1, \dots, \glab_n:\mtype_n, \dots\}::\kind{\glab_1:\mtype_1, \dots, \glab_n:\mtype_n} \\
          &\quad{} \text{if} \ \{\glab_1:\mtype_1, \dots, \glab_n:\mtype_n, \dots\} \ \text{is well formed under} \ \kenv
\end{align*}
\begin{itemize}  
\item  First, let us consider the case where $\ttype :: \gkind$ is of the form $\mtype :: \ukind$ and $\ttype$ is well formed under $\kenv$. We want to show that $\kenv_1 \Vdash \subs(\mtype) :: \subs(\ukind)$, \textit{i.e.} $\kenv_1 \Vdash \subs(\mtype) :: \ukind$. We can do this by induction on the structure of $\mtype$. 
  \begin{itemize}
  \item $\mtype = \btype$. Trivial.
  \item $\mtype = \vtype$. We want to show that $\kenv_1 \Vdash \subs(\vtype) :: \subs(\ukind)$. We know that $\kenv_1(\vtype) = \ukind$ and that $(\kenv_1, \subs)$ respects $\kenv_1$. Now, if $\vtype \in \dom{\subs}$, then we know that $\subs(\vtype)$ is well formed under $\kenv_1$, and, therefore, that $\kenv_1 \Vdash \subs(\vtype) :: \subs(\ukind)$, if $\vtype \not\in \dom{\subs}$, then, since $\vtype \in \dom{\kenv}$, and we know that $\kenv_1 \Vdash \vtype :: \ukind$, then we have that $\kenv_1 \Vdash \subs(\vtype) :: \ukind$.
  \item  $\mtype = \{\glab_1 : \mtype_1, \dots, \glab_n : \mtype_n\}$. Since $\{\glab_1 : \mtype_1, \dots, \glab_n : \mtype_n\}$ is well formed under $\kenv$, then we know that $\ftv{\{\glab_1 : \mtype_1, \dots, \glab_n : \mtype_n\}} = \ftv{\mtype_1} \cup \cdots \cup \ftv{\mtype_n} \subseteq \dom{\kenv}$. By Lemma \ref{lem0}, we have that $\ftv{\subs(\mtype_1)} \cup \cdots \cup \ftv{\subs(\mtype_n)} \subseteq \dom{\kenv_1}$, \textit{i.e.} $\ftv{\subs(\{\glab_1 : \mtype_1, \dots, \glab_n : \mtype_n\})} \subseteq \dom{\kenv_1}$, therefore $\subs(\{\glab_1 : \mtype_1, \dots, \glab_n : \mtype_n\})$ is well formed under $\kenv_1$, and $\kenv_1 \Vdash \subs(\{\glab_1 : \mtype_1, \dots, \glab_n : \mtype_n\}) :: \ukind$. 
  \item $\mtype = \atype{\mtype_1}{\mtype_2}$. We want to show that $\subs(\atype{\mtype_1}{\mtype_2})$ is well formed under $\kenv_1$, and, therefore, $\kenv_1 \Vdash \subs(\atype{\mtype_1}{\mtype_2}) :: \ukind$. We know that $\kenv \Vdash \atype{\mtype_1}{\mtype_2} :: \ukind$ and $(\kenv_1, \subs)$ respects $\kenv$. By the induction hypothesis, we know that $\kenv \Vdash \mtype_1 :: \gkind_1$, \textit{i.e.} $\kenv_1 \Vdash \subs(\mtype_1) :: \subs(\gkind_1)$ and $\kenv \Vdash \mtype_2 :: \gkind_2$, \textit{i.e.} $\kenv_1 \Vdash \subs(\mtype_2) :: \subs(\gkind_2)$, therefore, both $\subs(\mtype_1)$ and $\subs(\mtype_2)$ are well formed under $\kenv_1$, which, in turn, means that $\ftv{\subs(\mtype_1)} \subseteq \dom{\kenv_1}$ and $\ftv{\subs(\mtype_2)} \subseteq \dom{\kenv_1}$. But, then $\ftv{\subs(\mtype)} \cup \ftv{\subs(\mtype_2)} \subseteq \dom{\kenv_1}$, \textit{i.e.} $\ftv{\subs(\atype{\mtype_1}{\mtype_2})} \subseteq \dom{\kenv_1}$, which means that $\subs(\atype{\mtype_1}{\mtype_2})$ is well formed under $\kenv_1$, and, therefore, $\kenv_1 \Vdash \subs(\atype{\mtype_1}{\mtype_2}) :: \ukind$. 
\end{itemize}
  \item $\ttype :: \gkind$ is of the form $\vtype :: \kind{\glab_1 : \mtype_1, \dots, \glab_n : \mtype_n}$, and $\kenv(\vtype) = \kind{\glab_1 : \mtype_1, \dots, \glab_n : \mtype_n, \dots}$. We want to show that $\kenv_1 \Vdash \subs(\vtype) :: \subs(\kind{\glab_1 : \mtype_1, \dots, \glab_n : \mtype_n})$. Since $(\kenv_1, \subs)$ respects $\kenv$, we know that $\forall \vtype' \in \dom{\kenv}$, $\kenv_1 \Vdash \subs(\vtype') :: \subs(\kenv(\vtype'))$. But, since we know that $\vtype \in \dom{\kenv}$ and $\kenv(\vtype) = \kind{\glab_1 : \mtype_1, \dots, \glab_n : \mtype_n, \dots}$, then, if we take $\vtype' = \vtype$, we have that $\kenv_1 \Vdash \subs(\vtype) :: \subs(\kind{\glab_1 : \mtype_1, \dots, \glab_n : \mtype_n})$. 
  \item $\ttype :: \gkind$ is of the form $\{\glab_1 : \mtype_1, \dots, \glab_n : \mtype_n, \dots\} :: \kind{\glab_1 : \ltype_1, \dots, \glab_n : \mtype_n}$ and $\{\glab_1 : \mtype_1, \dots, \glab_n : \mtype_n, \dots\}$ is well formed under $\kenv$. We want to show that $\kenv_1 \Vdash \subs(\{\glab_1 : \mtype_1, \dots, \glab_n : \mtype_n, \dots\}) :: \subs(\kind{\glab_1 : \mtype_1, \dots, \glab_n : \mtype_n})$, which only happens if $\subs(\{\glab_1 : \mtype_1, \dots, \glab_n : \mtype_n, \dots\})$ is well formed under $\kenv_1$. Since $\{\glab_1 : \mtype_1, \dots, \glab_n : \mtype_n, \dots\}$ is well formed under $\kenv$, we know that $\ftv{\{\glab_1 : \mtype_1, \dots, \glab_n : \mtype_n, \dots\}} = \ftv{\mtype_1} \cup \cdots \cup \ftv{\mtype_n} \cup \cdots \subseteq \dom{\kenv}$. Now, by Lemma \ref{lem0} we have that $\ftv{\subs(\mtype_1)} \cup \dots \cup \ftv{\subs(\mtype_n)} \cup \dots \subseteq \dom{\kenv_1}$, \textit{i.e.} $\ftv{\subs(\{\glab_1 : \mtype_1, \dots, \glab_n : \mtype_n, \dots\})} \subseteq \dom{\kenv_1}$, and, therefore, $\subs(\{\glab_1 : \mtype_1, \dots, \glab_n : \mtype_n, \dots\})$ is well formed under $\kenv_1$, and $\kenv_1 \Vdash \subs(\{\glab_1 : \mtype_1, \dots, \glab_n : \mtype_n, \dots\}) :: \subs(\kind{\glab_1 : \mtype_1, \dots, \glab_n : \mtype_n})$.
  \end{itemize}
  \vspace{-0.3in}
\end{proof}

\begin{definition}\label{def:geninst}
Let $\ttype_1$ be a well-formed type under a kinding environment $\kenv$. Then, $\ttype_2$ is a generic instance of $\ttype_1$ under $\kenv$ (denoted as $\kenv \Vdash \ttype_1 \ge \ttype_2$), if $\ttype_1 = \forall \vtype_1::\gkind^1_1 \cdots \forall \vtype_n::\gkind^1_n.\mtype_1$, $\ttype_2 = \forall\beta_1::\gkind^2_1 \cdots \forall\beta_m::\gkind^2_m.\mtype_2$, and there exists a substitution $\subs$ such that $\dom{\subs} = \{\vtype_1, \dots, \vtype_n\}$, $(\kenv \cup \{\beta_1::\gkind^2_1, \dots, \beta_m::\gkind^2_m\}, \subs)$ respects $\kenv \cup \{\vtype_1::\gkind^1_1, \dots, \vtype_n::\gkind^1_n\}$, and $\mtype_2 = \subs(\mtype_1)$.
\end{definition}



\begin{definition}\label{def:closure}
  Let $\tenv$ be a typing environment and $\mtype$ be a type, both well-formed under a kinding environment $\kenv$. The closure of $\mtype$ under $\tenv$ and $\kenv$ (denoted as $\cls{\kenv}{\tenv}{\mtype}$) is a pair $(\kenv', \forall \vtype_1::\gkind_1 \cdots \forall \vtype_n::\gkind_n.\mtype)$ such that $\kenv' \cup \{\vtype_1 :: \gkind_1, \dots \vtype_n :: \gkind_n\} = \kenv$ and $\{\vtype_1, \dots, \vtype_n\} = \eftv{\kenv}{\mtype} \setminus \eftv{\kenv}{\tenv}$.
\end{definition}

\begin{example}
Let $\kenv = \{\vtype_2::\ukind, \vtype_3::\ukind, \vtype_4::\ukind, \vtype_1::\kind{l_1:\vtype_2}\}$, $\tenv = \{\vterm : \vtype_1\}$, and $\mtype = \atype{\{l_1:\vtype_2, l_4:\booltype\}}{\{l_2:\inttype, l_3: \ctype{\vtype_3}{\vtype_4}\}}$. Then $\cls{\kenv}{\tenv}{\mtype} = (\{\vtype_2::\ukind, \vtype_1::\kind{l_1:\vtype_2}\}, \forall \vtype_3::\ukind.\forall \vtype_4::\ukind.\atype{\{l_1:\vtype_2, l_4:\booltype\}}{\{l_2:\inttype, l_3:\ctype{\vtype_3}{\vtype_4}\}})$, because $\kenv = \{\vtype_2::\ukind, \vtype_1::\kind{l_1:\vtype_2}\} \cup \{\vtype_3::\ukind, \vtype_4::\ukind\}$, $\eftv{\kenv}{\mtype} = \{\vtype_2, \vtype_3, \vtype_4, \vtype_1\}$,  $\eftv{\kenv}{\tenv} = \{\vtype_1, \vtype_2\}$, and $\eftv{\kenv}{\mtype} \setminus \eftv{\kenv}{\tenv} = \{\vtype_2, \vtype_3, \vtype_4, \vtype_1\} \setminus \{\vtype_1, \vtype_2\} = \{\vtype_3, \vtype_4\}$.
\end{example}

The type assignment system for \EVL is given in Figure~\ref{fig:typesystem}, and can be seen as both a restriction and an extension of the Ohori type system for record types. Unlike Ohori, we do not deal with variant types in this system, but we have additional language constructors, like conditionals and explicit event definition. We use $\kenv,\tenv \vdash \gterm: \sigma$ to denote that the \EVL term $\gterm$ has type $\sigma$ given the type and kind environments $\tenv$ and $\kenv$, respectively.

\begin{figure}[htbp]
\begin{prooftree}
  \AxiomC{ \textit{$\tenv$ is well-formed under $\kenv$}}
  \RightLabel{(Const)}
  \UnaryInfC{$\kenv, \tenv \vdash \cterm^\btype : \btype$}
\end{prooftree}

\begin{prooftree}
  \AxiomC{$\kenv \Vdash \tenv(\vterm) \geq \mtype$, \textit{$\tenv$ is well-formed under $\kenv$}}
  \RightLabel{(Var)}
  \UnaryInfC{$\kenv, \tenv \vdash \vterm : \mtype$}
\end{prooftree}

\begin{prooftree}
  \AxiomC{$\kenv, \tenv \vdash \gterm_1 : \atype{\mtype_1}{\mtype_2}$}
  \AxiomC{$\kenv, \tenv \vdash \gterm_2 : \mtype_1$}
  \RightLabel{(App)}
  \BinaryInfC{$\kenv, \tenv \vdash \app{\gterm_1}{\gterm_2} : \mtype_2$}
\end{prooftree}

\begin{prooftree}
  \AxiomC{$\kenv, \tenv_x \cup \{\vterm : \mtype_1\} \vdash \gterm : \mtype_2$}
  \RightLabel{(Abs)}
  \UnaryInfC{$\kenv, \tenv_{\vterm} \vdash \abs{\vterm}{\gterm} : \atype{\mtype_1}{\mtype_2}$}
\end{prooftree}

\begin{prooftree}
  \AxiomC{$\kenv, \tenv \vdash \gterm_1 : \bltype$}
  \AxiomC{$\kenv, \tenv \vdash \gterm_2 : \mtype$}
  \AxiomC{$\kenv, \tenv \vdash \gterm_3 : \mtype$}
  \RightLabel{(Cond)}
  \TrinaryInfC{$\kenv, \tenv \vdash \cond{\gterm_1}{\gterm_2}{\gterm_3} : \mtype$}
\end{prooftree}

\begin{prooftree}
  \AxiomC{$\kenv', \tenv_{\vterm} \vdash \gterm_1 : \mtype'$}
  \AxiomC{$\cls{\kenv'}{\tenv_{\vterm}}{\mtype'} = (\kenv, \ttype)$}
  \AxiomC{$\kenv, \tenv_{\vterm} \cup \{\vterm : \ttype\} \vdash \gterm_2 : \mtype$}
  \RightLabel{(Let)}
  \TrinaryInfC{$\kenv, \tenv_{\vterm} \vdash \letw{\vterm}{\gterm_1}{\gterm_2} : \mtype$}
\end{prooftree}

\begin{prooftree}
  \AxiomC{$\kenv', \tenv_{\vterm} \vdash \gterm_1 : \etype$}
  \AxiomC{$\cls{\kenv'}{\tenv_{\vterm}}{\etype} = (\kenv, \ttype)$}
  \AxiomC{$\kenv, \tenv_{\vterm} \cup \{\vterm : \ttype\} \vdash \gterm_2 : \mtype$}
  \RightLabel{(LetEv)}
  \TrinaryInfC{$\kenv, \tenv_{\vterm} \vdash \letEv{\vterm}{\gterm_1}{\gterm_2} : \mtype$}
\end{prooftree}

\begin{prooftree}
  \AxiomC{$\kenv, \tenv \vdash \gterm_i : \mtype_i, 1 \le i \le n$}
  \RightLabel{(Rec)}
  \UnaryInfC{$\kenv, \tenv \vdash \{l_1 = \gterm_1, \dots, l_n = \gterm_n\} : \{l_1 : \mtype_1, \dots, l_n : \mtype_n\}, n \ge 1$}
\end{prooftree}

\begin{prooftree}
  \AxiomC{$\kenv, \tenv \vdash \gterm : \mtype'$}
  \AxiomC{$\kenv \Vdash \mtype' :: \kind{l:\mtype}$}
  \RightLabel{(Sel)}
  \BinaryInfC{$\kenv, \tenv \vdash \gterm.l : \mtype$}
\end{prooftree}

\begin{prooftree}
  \AxiomC{$\kenv, \tenv \vdash \gterm_1 : \mtype$}
  \AxiomC{$\kenv, \tenv \vdash \gterm_2 : \mtype'$}
  \AxiomC{$\kenv \Vdash \mtype :: \kind{l:\mtype'}$}
  \RightLabel{(Modif)}
  \TrinaryInfC{$\kenv, \tenv \vdash \modw{\gterm_1}{l}{\gterm_2} : \mtype$}
\end{prooftree}
\caption{Type assignment system for \EVL}
\label{fig:typesystem}
\end{figure}

\begin{example}
  \label{ex:typesystem}Let $\shterm = \{location = l, fire\_danger = d\}$, $\shtype_1 = \{location : \alpha_1, fire\_danger : \alpha_2\}$, $\shtype_2 = \forall\alpha_1 :: \ukind.\forall\alpha_2 :: \ukind.\atype{\alpha_1}{\atype{\alpha_2}{\shtype_1}}$, $\shtype_3 = \{location : \strtype, fire\_danger : \strtype\}$, and $\shtype_4 = \atype{\strtype}{\atype{\strtype}{\shtype_3}}$.
  In Figure~\ref{fig:typederiv} we give a type derivation for: $$\letEv{FireDanger}{\abs{l}{\abs{d}{\shterm}}}{FireDanger \ \porto^\strtype \ \low^\strtype}.$$
  
\end{example}
\begin{figure}[htbp]
{\footnotesize
 \begin{prooftree}
  \AxiomC{}
  \RightLabel{(Var)}
  \UnaryInfC{$\{\alpha_1 :: \ukind, \alpha_2 :: \ukind\}, \{l : \alpha_1, d : \alpha_2\} \ \cup \ \tenv \vdash l : \alpha_1$}
  \AxiomC{}
  \RightLabel{(Var)}
  \UnaryInfC{$\{\alpha_1 :: \ukind, \alpha_2 :: \ukind\}, \{l : \alpha_1, d : \alpha_2\} \ \cup \ \tenv \vdash d : \alpha_2$}
  \RightLabel{(Rec)}
  \BinaryInfC{$\{\alpha_1 :: \ukind, \alpha_2 :: \ukind\}, \{l : \alpha_1, d : \alpha_2\} \ \cup \ \tenv \vdash \shterm : \shtype_1$}
  \RightLabel{(Abs)}
  \UnaryInfC{$\{\alpha_1 :: \ukind, \alpha_2 :: \ukind\}, \{l : \alpha_1\} \ \cup \ \tenv \vdash \abs{d}{\shterm} : \atype{\alpha_2}{\shtype_1}$}
  \RightLabel{(Abs)}
 \LeftLabel{$\Phi_1 = $}
  \UnaryInfC{$\{\alpha_1 :: \ukind, \alpha_2 :: \ukind\}, \tenv \vdash \abs{l}{\abs{d}{\shterm}} : \atype{\alpha_1}{\atype{\alpha_2}{\shtype_1}}$}
\end{prooftree}

\begin{prooftree}
    \AxiomC{$\Phi_2 = \cls{\{\alpha_1 :: \ukind, \alpha_2 :: \ukind\}}{\tenv}{\atype{\alpha_1}{\atype{\alpha_2}{\shtype_1}}} = (\{\}, \shtype_2)$}
\end{prooftree}

\begin{prooftree}
    \AxiomC{}
    \RightLabel{(Var)}
    \UnaryInfC{$\{\}, \{FireDanger : \shtype_2\} \vdash FireDanger : \shtype_4$}
    \AxiomC{}
    \RightLabel{(Const)}
    \UnaryInfC{$\{\}, \{FireDanger : \shtype_2\} \vdash \porto^\strtype : \strtype$}
    \RightLabel{(App)}
    \LeftLabel{$\Phi_3 =$}
    \BinaryInfC{$\{\}, \{FireDanger : \shtype_2\} \vdash FireDanger \ \porto^\strtype : \atype{\strtype}{\shtype_3}$}
\end{prooftree}

\begin{prooftree}
    \AxiomC{$\Phi_3$}
    \AxiomC{}
    \RightLabel{(Const)}
    \UnaryInfC{$\{\}, \{FireDanger : \shtype_2\} \vdash \low^\strtype : \strtype$}
    \RightLabel{(App)}
    \LeftLabel{$\Phi_4 =$}
    \BinaryInfC{$\{\}, \{FireDanger : \shtype_2\} \ \cup \ \tenv \vdash FireDanger \ ``Porto" \ ``low" : \shtype_3$}
\end{prooftree}

\begin{prooftree}
    \AxiomC{$\Phi_1$}
    \AxiomC{$\Phi_2$}
    \AxiomC{$\Phi_4$}
    \RightLabel{(LetEv)}
    \TrinaryInfC{$\{\}, \tenv \vdash \letEv{FireDanger}{\abs{l}{\abs{d}{\shterm}}}{FireDanger \ ``Porto" \ ``low"} : \shtype_3$}
\end{prooftree}

}
\caption{A type derivation for the \EVL term $\letEv{FireDanger}{\abs{l}{\abs{d}{M}}}{FireDanger \ \porto^\strtype \ \low^\strtype}$
}
\label{fig:typederiv}
\end{figure}

\begin{lemma}\label{lem2}
  If $\kenv, \tenv \vdash \gterm : \ttype$ and $(\kenv_1, \subs)$ respects $\kenv$, then $\kenv_1, \subs(\tenv) \vdash \gterm : \subs(\ttype)$.
\end{lemma}
\begin{proof}
 The proof is by induction on the typing derivations.
 \begin{itemize}
 \item (Const) Trivial.
 \item (Var) Then we know that $\tenv$ is well formed under $\kenv$ and $(\kenv_1,\subs)$ respects $\kenv_1$, so $\subs(\tenv)$ is well formed under $\kenv_1$. Now, since $\subs(\tenv)(\vterm) = \subs(\ttype)$, we can conclude that $\kenv_1, \subs(\tenv) \vdash \vterm : \subs(\ttype)$ by (Var). 
 \item (App) Then we know that $\kenv, \tenv \vdash \gterm_1 \gterm_2 : \mtype_2$, for some type $\mtype_2$, and, therefore, that $\kenv, \tenv \vdash \gterm_1 : \atype{\mtype_1}{\mtype_2}$ and $\kenv, \tenv \vdash \gterm_2 : \mtype_1$ must exist, for some type $\mtype_2$. Now, by the induction hypothesis, we know that $\kenv_1, \subs(\tenv) \vdash \gterm_1 : \subs(\atype{\mtype_1}{\mtype_2})$, \textit{i.e.} $\kenv_2, \subs(\tenv) \vdash \gterm_1 : \atype{\subs(\mtype_1)}{\subs(\mtype_2)}$, and $\kenv_2, \subs(\tenv) \vdash \gterm_2 : \subs(\mtype_1)$. Finally, by (App), we have that $\kenv_1, \subs(\tenv) \vdash \gterm_1 \gterm_2 : \subs(\mtype_2)$. 
\item (Abs) Then we know that $\kenv, \tenv \vdash \gterm : \atype{\mtype_1}{\mtype_2}$ for some types $\mtype_1$ and $\mtype_2$, and, therefore, that $\kenv, \tenv \cup \{\vterm : \mtype_1\} \vdash \gterm : \mtype_2$ must exist. Now, by the induction hypothesis, we know that $\kenv_1, \subs(\tenv \cup \{\vterm : \mtype_1\}) \vdash \gterm : \subs(\mtype_2)$, \textit{i.e.} $\kenv_1, \subs(\tenv) \cup \{\vterm : \subs(\mtype_1)\} \vdash \gterm : \subs(\mtype_2)$. Finally, by (Abs), we have that $\kenv_2, \subs(\tenv) \vdash \lambda \vterm.\gterm : \atype{\subs(\mtype_1)}{\subs(\mtype_2)}$, \textit{i.e.} $\kenv_1, \subs(\tenv) \vdash \lambda \vterm.\gterm : \subs(\atype{\mtype_1}{\mtype_2})$. 
    \item (Let) Then we know that $\kenv, \tenv \vdash \letw{\vterm}{\gterm_1}{\gterm_2} : \mtype$, for some type $\mtype$, and, therefore, that $\kenv', \tenv \vdash \gterm_1 : \mtype'$, $\cls{\kenv'}{\tenv}{\mtype'} = (\kenv, \ttype)$, and $\kenv, \tenv_{x} \cup \{x : \ttype\} \vdash \gterm_2 : \mtype$. By the induction hypothesis, we know that $\kenv_1, \subs(\tenv_{x} \cup \{x : \ttype\}) \vdash \gterm_2 : \subs(\mtype)$. By the definition of $\textit{Cls}$, we know we can write $\kenv' = \kenv \cup \{\alpha_1 :: \gkind_1, \alpha_n :: \gkind_n\}$ such that $\{\alpha_1, \dots, \alpha_2\} = \eftv{\kenv'}{\mtype'} \setminus \eftv{\kenv'}{\tenv}$ and $\ttype = \forall\alpha_1::\gkind_1\cdots\forall\alpha_n::\gkind_n.\mtype'$. Now, let $\kenv'' = \kenv_1 \cup \{\alpha_1 :: \subs'(\gkind_1), \dots, \alpha_n :: \subs'(\gkind_n)\}$, where $\subs'$ is the restriction of $\subs$ on $\dom{S}\setminus\{\alpha_1, \dots, \alpha_n\}$. Since $(\kenv_1, \subs)$ respects $\kenv$ and $\kenv' = \kenv \cup \{\alpha_1 :: \subs'(\gkind_1), \dots, \alpha_n :: \subs'(\gkind_n)\}$, then $(\kenv'', \subs')$ respects $\kenv'$. By the induction hypothesis, we have that $\kenv'', \subs'(\tenv) \vdash \gterm_1 : \subs'(\mtype')$. Again, by the definition of $\textit{Cls}$, we obtain $\cls{\kenv''}{\subs'(\tenv)}{\subs'(\mtype')} = (\kenv_1, \forall\alpha_1::\subs'(\gkind_1)\cdots\forall\alpha_n::\subs'(\gkind_n).\subs'(\mtype'))$ and it is easy to see that $\subs(\ttype) = \subs'(\ttype)$ and $\subs(\tenv) = \subs'(\tenv)$. Finally, by (Let), we have that $\kenv_1, \subs(\tenv) \vdash \gterm_2 : \subs(\mtype)$.
\item  (LetEv) Identical to the case for (Let). 
\item (Rec) Then we know that $\kenv, \tenv \vdash \{\glab_1 = \gterm_1, \dots, \glab_n = \gterm_n\} : \{\glab_1 : \mtype_1, \dots, \glab_n : \mtype_n\}$, and, therefore, that $\kenv, \tenv \vdash \gterm_i : \mtype_i, (1 \le i \le n)$, for some types $\mtype_i$. Now, by the induction hypothesis, we know that $\kenv_1, \subs(\tenv) \vdash \gterm_i : \subs(\mtype_i), (1 \le i \le n)$. Finally, by (Rec), we have that $\kenv_1, \subs(\tenv) \vdash \{\glab_1 = \gterm_1, \dots, \glab_n = \gterm_n\} : \{\glab_1 : \subs(\mtype_1), \dots, \glab_n : \subs(\mtype_n)\}$, \textit{i.e.} $\kenv_1, \subs(\tenv) \vdash \{\glab_1 = \gterm_1, \dots, \glab_n = \gterm_n\} : \subs(\{\glab_1 : \mtype_1, \dots, \glab_n : \mtype_n\})$. 
\item (Sel) Then we know that $\kenv, \tenv \vdash \sel{\gterm}{\glab} : \mtype'$, for some type $\mtype'$, and, therefore, that $\kenv, \tenv \vdash \gterm : \mtype$, must exist, for some type $\mtype$, and $\kenv \Vdash \mtype :: \kind{\glab : \mtype'}$. Now, by the induction hypothesis, we know that $\kenv_1, \subs(\tenv) \vdash \gterm : \subs(\mtype)$, and, by lemma \ref{lem1}, we know that, since $\kenv \Vdash \mtype :: \kind{\glab : \mtype'}$ and $(\kenv_1, \subs)$ respects $\kenv$, $\kenv_1 \Vdash \subs(\mtype) :: \subs(\kind{\glab : \mtype'})$, \textit{i.e.} $\kenv_1 \Vdash \subs(\mtype) :: \kind{\glab : \subs(\mtype')}$. Finally, by (Sel), we have that $\kenv_1, \subs(\tenv) \vdash \sel{\gterm}{\glab} : \subs(\mtype')$. 
\item (Modify) Then we know that $\kenv, \tenv \vdash \modw{\gterm_1}{\glab}{\gterm_2} : \mtype$, for some type $\mtype$, and, therefore, that $\kenv, \tenv \vdash \gterm_1 : \mtype$, $\kenv, \tenv \vdash \gterm_2 : \mtype'$, for some type $\mtype'$, and $\kenv \Vdash \mtype :: \kind{\glab : \mtype'}$. Now, by the induction hypothesis, we know that $\kenv_1, \subs(\tenv) \vdash \gterm_1 : \subs(\mtype)$ and $\kenv_1, \subs(\tenv) \vdash \gterm_2 : \subs(\mtype')$, and, by Lemma \ref{lem1}, we know that, since $\kenv \Vdash \mtype :: \kind{\glab : \mtype'}$ and $(\kenv_1, \subs)$ respects $\kenv$, $\kenv_1 \Vdash \subs(\mtype) :: \subs(\kind{\glab : \mtype'})$, \textit{i.e.} $\kenv_1 \Vdash \subs(\mtype) :: \kind{\glab : \subs(\mtype')}$. Finally, by (Modif), we have that $\kenv_1, \subs(\tenv) \vdash \modw{\gterm_1}{\glab}{\gterm_2} : \mtype$.
\item (Cond) Then $\kenv, \tenv \vdash \cond{\gterm_1}{\gterm_2}{\gterm_3} : \mtype$, for some type $\mtype$, and, therefore, $\kenv, \tenv \vdash \gterm : \bltype$, $\kenv, \tenv \vdash \gterm_2 : \mtype$ and $\kenv, \tenv \vdash \gterm_3 : \mtype$ must exist. Now, by the induction hypothesis, we know that $\kenv_1, \subs(\tenv) \vdash \gterm_1 : \subs(\bltype)$, \textit{i.e.} $\kenv_1, \subs(\tenv) \vdash \gterm_1 : \bltype$, $\kenv_1, \subs(\tenv) \vdash \gterm_2 : \subs(\mtype)$ and $\kenv_1, \subs(\tenv) \vdash \gterm_3 : \mtype$. Finally, by (Cond), we have that $\kenv_1, \subs(\tenv) \vdash \cond{\gterm_1}{\gterm_2}{\gterm_3} : \mtype$.
\end{itemize}
\vspace{-0.3in}
\end{proof}

\section{Operational semantics for \EVL}
\label{sec:cbv}
We define a call-by-value operational semantics for \EVL using  evaluation contexts~\cite{FELLEISEN1987205}. This semantics is based on Ohori's operational semantics for his ML-style record calculus and serves as an evaluation model of a polymorphic programming language with records.

\subsection{Operational semantics}
The set of values (ranged over by $\val$) is given by the following grammar:
\begin{align*}
    \val ::= \cterm^\btype \mid \abs{\vterm}{\gterm} \mid \{\glab = \val, \dots, \glab = \val\} \\
\end{align*}
Evaluation contexts guide the evaluation of  terms. The set of evaluation contexts (ranged over by $\eval{\hole}$) is given by the following grammar, where $[\bullet]$ represents the empty context ($\bullet$ is called a hole):

\[
\begin{array}{lcl}
    \eval{\hole} & ::= & [\bullet] \mid \app{\eval{\hole}}{\gterm} \mid  \app{\val}{\eval{\hole}} \mid \cond{\eval{\hole}}{\gterm_1}{\gterm_2} \\ 
    & & \letw{\vterm}{\eval{\hole}}{\gterm} \mid \letEv{\vterm}{\eval{\hole}}{\gterm} \\
    & & \{\glab_1 = \val_1, \dots, \glab_{i-1} = \val_{i-1}, \glab_i = \eval{\hole}, \dots\} \mid \sel{\eval{\hole}}{\glab} \mid \modw{\eval{\hole}}{\glab}{\gterm} \mid \modw{\val}{\glab}{\eval{\hole}} \\ 
\end{array}
\]

Let $\eval{\gterm}$ be the term obtained by placing $\gterm$ in the hole of the context $\eval{\hole}$.
The set of call-by-value context-rewriting axioms are given by the following rules: 
\[
\begin{array}{lcl}
    \eval{(\abs{\vterm}{\gterm}) \ \val}  & \rightarrow & \eval{[\val/\vterm]\gterm} \\
    \eval{\cond{\true}{\gterm_1}{\gterm_2}} & \rightarrow & \eval{\gterm_1} \\
    \eval{\cond{\false}{\gterm_1}{\gterm_2}} & \rightarrow & \eval{\gterm_2} \\
    \eval{\letw{\vterm}{\val}{\gterm}} & \rightarrow & \eval{[\val/\vterm]\gterm} \\
    \eval{\letEv{\vterm}{\val}{\gterm}} & \rightarrow & \eval{[\val/\vterm]\gterm} \\
    \eval{\sel{\{\glab_1 = \val_1, \dots, \glab_n = \val_n\}}{\glab_i}} & \rightarrow & \eval{\val_i} \\
    \eval{\modw{\{\glab_1 = \val_1, \dots, \glab_n = \val_n\}}{\glab_i}{\val}} & \rightarrow & \eval{\{\glab_1 = \val_1, \dots, \glab_i = \val, \dots, \glab_n = \val_n\}}
\end{array}
\]

A one-step evaluation relation $\gterm \xrightarrow[]{\evall} \gterm'$ is then defined as: there exist $\eval{\hole}, \gterm_1, \gterm_2$ such that $\gterm = \eval{\gterm_1}, \eval{\gterm_1} \rightarrow \eval{\gterm_2}$, and $\gterm' = \eval{\gterm_2}$. We write $\gterm \xrightarrow[]{\evall}^{*} \gterm'$ for the reflexive and transitive closure of $\xrightarrow[]{\evall}$ and we write $\gterm \downarrow \gterm'$ if $\gterm \xrightarrow[]{\evall}^{*} \gterm'$ and there is no $\gterm''$ such that $\gterm' \xrightarrow[]{\evall}^{*} \gterm''$.

\subsection{Type soundness}
We now show type soundness with respect to this operational semantics. We start by defining a type-indexed family of predicates on closed values.

For a closed type $\ttype$, let $\vvalue^{\ttype} = \{\val \mid \emptyset, \emptyset \vdash \val : \ttype\}$ and define $\pred^{\ttype} \subseteq \vvalue^{\ttype}$ by induction on $\ttype$ as follows:

\[
\begin{array}{lcl}
    -\ \val \in \pred^{\btype} \ \text{iff} \ \val = \cterm^\btype \ \text{for some constant} \ \cterm^\btype \\
    -\ \val \in \pred^{\forall\vtype_1::\gkind_1 \cdots \forall\vtype_n::\gkind_n.\mtype} \ \text{iff for any ground substitution} \ \subs \ \text{such that} \ \dom{\subs} = \{\vtype_1, \dots, \vtype_n\} \\
    \text{and} \ \subs \ \text{satisfies} \ \{\vtype_1 :: \gkind_1, \dots, \vtype_n :: \gkind_n\}, \val \in \pred^{\subs(\mtype)} \\
    -\ \val \in \pred^{\atype{\mtype_1}{\mtype_2}} \ \text{iff for any} \ \val_0 \in \pred^{\mtype_1}, \ \text{if} \ (\app{\val}{\val_0}) \downarrow \gterm \ \text{then} \ \gterm \in \pred^{\mtype_2} \\
    -\ \val \in \pred^{\{\glab_1 : \mtype_1, \dots, \glab_n : \mtype_n\}} \ \text{iff} \ \val = \{\glab_1 = \val_1, \dots, \glab_n = \val_n\} \ \text{such that} \ \val_i \in \pred^{\mtype_i} (1 \leq i \leq n) \\
\end{array}
\]

\begin{definition}
    Let $\tenv$ be a closed type assignment. A $\tenv$-\textit{environment} is a function $\eta$ such that $\dom{\eta} = \dom{\tenv}$ and for any $\vterm \in \dom{\tenv}$, $\eta(\vterm) \in \vvalue^{\tenv(\vterm)}$ $\eta$ is a function that replaces each (bound) variable $\vterm$ of a closed term, a value with type $\tenv(\vterm)$.
\end{definition}

If $\eta$ is an environment, we write $\eta(\gterm)$ for the term obtained from $\gterm$ by substituting $\eta(\vterm)$ for each free occurrence of $\vterm$ in $\gterm$. For a function $f$, if $\vterm \not\in \dom{f}$, then we write $f\{\vterm \mapsto \val\}$ for the extension $f'$ of $f$ to $\vterm$ such that $f'(x) = \val$.

\begin{example}
    Let $\kenv = \{\vtype_1 :: \ukind, \vtype_2 :: \kind{\glab_1 : \vtype_1}\}$ and $\tenv = \{x_1 : \vtype_1, x_2 : \vtype_2\}$. Then $\kenv, \tenv \vdash \sel{x_2}{\glab_1} : \vtype_1$ and $\subs = [\booltype/\vtype_1, \{\glab_1 : \vtype_1, \glab_2 : \vtype_1\}/\vtype_2]$ respects $\kenv$.
    
    Consider the following $\subs(\tenv)$-environment:
      \[
        \eta : \{x_1, x_2\}  \mapsto  \booltype,\ x_1  \mapsto  \true^\booltype,\ x_2 \mapsto  \{\glab_1 = \true^\booltype, \glab_2 = \false^\booltype\}.\]
    Then $\eta(\sel{x_2}{\glab_1}) = \sel{\{\glab_1 = \true^\booltype, \glab_2 = \false^\booltype\}}{\glab_1}$, $\sel{\{\glab_1 = \true^\booltype, \glab_2 = \false^\booltype\}}{\glab_1} \downarrow \true^\booltype$ and $\true^\booltype \in \pred^{\subs(\alpha_1)} = \pred^\booltype$.
\end{example}

\begin{therm}
    \label{thm:typepreservation}
    If $\kenv, \tenv \vdash \gterm : \ttype$ then for any ground substitution $\subs$ that respects $\kenv$, and for any $\subs(\tenv)$-\textit{environment} $\eta$, if $\eta(\gterm) \downarrow \gterm'$, then $\gterm' \in \pred^{\subs(\ttype)}$.
\end{therm}

\begin{proof}
    Let $\subs$ be any ground substitution respecting $\kenv$, and let $\eta$ be any $\subs(\tenv)$-\textit{environment}. We proceed by induction on the typing derivation. 
    \begin{itemize}
        \item Case (Const): Trivial.
        \item Case (Var): Suppose $\kenv, \tenv \vdash \vterm : \mtype$. Then $\kenv, \tenv \Vdash \tenv(\vterm) \geq \mtype$. Let $\forall\vtype_1::\gkind_1\cdots\forall\vtype_n::\gkind_n.\mtype_0 = \tenv(x)$. Then there is some $\subs_0$ such that $\dom{\subs_0} = \{\vtype_1, \dots, \vtype_n\}$, $\tau = \subs_0(\mtype_0)$, and $\kenv \vdash \subs_0(\vtype_i) :: \subs_0(\gkind_i)$. By Lemma~\ref{lem1}, $\emptyset \Vdash \subs(\subs_0(\vtype_i)) :: \subs(\subs_0(\gkind_i))$. By the bound type variable convention, $\subs(\subs_0(\mtype_0)) = (\subs \circ \subs_0)(\subs(\mtype_0))$ and $\subs(\subs_0(\gkind_i)) = (\subs \circ \subs_0)(\subs(\gkind_i))$, since $\subs_0$ only affects bound variables. This means that $\subs \circ \subs_0$ is a ground substitution respecting $\{\vtype_1 :: \subs(\gkind_1), \dots, \vtype_n :: \subs(\gkind_n)\}$. Now, suppose that $\eta(\vterm) \downarrow \gterm'$. Then by the assumption $\gterm' \in \pred^{\forall\vtype_1::\subs(\gkind_1)\cdots\forall\vtype_n::\subs(\gkind_n).\subs(\mtype_0)}$. By the definition of $\pred$, we have that $\gterm' \in \pred^{(\subs \circ \subs_0)(\subs(\mtype_0))} = \pred^{\subs(\subs_0(\mtype_0))} = \pred^{\subs(\mtype)}$.
        \item Case (App): Suppose $\kenv, \tenv \vdash \app{\gterm_1}{\gterm_2} : \mtype_2$ is derived from $\kenv, \tenv \vdash \gterm_1 : \atype{\mtype_1}{\mtype_2}$ and $\kenv, \tenv \vdash \gterm_2 : \mtype_1$. Now, also suppose that $\eta(\app{\gterm_1}{\gterm_2}) \downarrow \gterm'$. By the definition of evaluation contexts, $\eta(\gterm_1) \downarrow \gterm_1'$ and $(\app{\gterm_1'}{\eta(\gterm_2)}) \downarrow \gterm$, since $(\app{\gterm_1}{\gterm_2})$ fits $(\app{\eval{\hole}}{\gterm})$ and $(\app{\gterm_1'}{\gterm_2})$ fits $(\app{\val}{\eval{\hole}})$. By the induction hypothesis for $\gterm_1$, we have that $\gterm_1' = \val_1 \in \pred^{\atype{\subs(\mtype_1)}{\subs(\mtype_2)}}$ for some value $\val_1$. But, by the definition of evaluation contexts, $\eta(\gterm_2) \downarrow \gterm_2'$ and $(\app{\val_1}{\gterm_2'}) \downarrow \gterm'$ and, by the induction hypothesis for $\gterm_2$, we have that $\gterm_2' = \val_2 \in \pred^{\subs(\mtype_1)}$ for some value $\val_2$. By the definition of $\pred$, we have $\gterm' \in \pred^{\subs(\mtype_2)}$.
        \item Case (Abs): Suppose $\kenv, \tenv \vdash \abs{\vterm}{\gterm_1} : \atype{\mtype_}{\mtype_2}$ is derived from $\kenv, \tenv\cup\{\vterm : \mtype_1\} \vdash \gterm_1 : \mtype_2$. Then, $\eta(\abs{\vterm}{\gterm_1}) = \abs{\vterm}{\eta(\gterm_1)} \downarrow \abs{\vterm}{\eta(\gterm_1)}$. This means that, if we want to see what happens to the type of $\abs{\vterm}{\gterm_1}$ during evaluation, we have to apply it to a value of type $\subs(\mtype_1)$. Let $\val$ be any element in $\pred^{\subs(\mtype_1)}$ and suppose $\app{(\abs{\vterm}{\eta(\gterm_1)})}{\val} \downarrow \gterm'$. By the definition of evaluation contexts, $[\val/\vterm](\eta(\gterm_1)) \downarrow \gterm'$, i.\,e., $\eta\{\vterm \mapsto \val\}(\gterm_1) \downarrow \gterm'$. Since $\eta\{\vterm \mapsto \val\}$ is a $\subs(\tenv\cup\{\vterm : \mtype_1\})$-\textit{environment}, by the induction hypothesis, $\gterm' \in \pred^{\subs(\mtype_2)}$. By the definition of $\pred$, this proves that $\abs{\vterm}{\eta(\gterm_1)} \in \pred^{\atype{\subs(\mtype_1)}{\subs(\mtype_2)}}$.
        \item Case (Cond): Suppose $\kenv, \tenv \vdash \cond{\gterm_1}{\gterm_2}{\gterm_3} : \mtype$ is derived from $\kenv, \tenv \vdash \gterm_1 : \bltype$, $\kenv, \tenv \vdash \gterm_2 : \mtype$, and $\kenv, \tenv \vdash \gterm_2 : \mtype$. Now, also suppose that $\eta(\cond{\gterm_1}{\gterm_2}{\gterm_3}) \downarrow \gterm'$. By the definition of evaluation contexts, $\eta(\gterm_1) \downarrow \gterm_1'$ and $(\cond{\gterm_1'}{\eta(\gterm_2)}{\eta(\gterm_3)}) \downarrow \gterm'$. By the induction hypothesis for $\gterm_1$, $\gterm_1' = \val_1 \in \pred^{\bltype}$, for some value $\val_1 \in \{\true^{\bltype}, \false^{\bltype}\}$. By the definition of evaluation contexts, $\eta(\gterm_2) \downarrow \gterm_2'$ and $\eta(\gterm_3) \downarrow \gterm_3'$. By the induction hypothesis for $\gterm_2$, $\gterm_2' = \val_2 \in \pred^{\subs(\mtype)}$, for some $\val_2$. And, by the induction hypothesis for $\gterm_3$, $\gterm_3' = \val_3 \in \pred^{\subs(\mtype)}$, for some $\val_3$. If $\val_1 = \true^{\bltype}$, then $(\cond{\gterm_1'}{\eta(\gterm_2)}{\eta(\gterm_3)}) \downarrow \val_2$. If $\val_1 = \false^{\bltype}$, then $(\cond{\gterm_1'}{\eta(\gterm_2)}{\eta(\gterm_3)}) \downarrow \val_3$.
        \item Case (Let): Suppose $\kenv, \tenv_\vterm \vdash \letw{\vterm}{\gterm_1}{\gterm_2} : \mtype$ is derived from $\kenv', \tenv_\vterm \vdash \gterm_1 : \mtype'$, $\cls{\kenv'}{\tenv_\vterm}{\mtype'} = (\kenv, \ttype)$, and $\kenv, \tenv_\vterm \cup \{\vterm : \ttype\} \vdash \gterm_2 : \mtype$. Then, there are some $\vtype_1, \dots, \vtype_n$ and $\gkind_1, \dots, \gkind_n$ such that $\kenv = \kenv' \cup \{\vtype_1 :: \gkind_1, \dots, \vtype_n :: \gkind_n\}$ and $\ttype = \forall\vtype_1::\gkind_1\cdots\forall\vtype_n::\gkind_n.\mtype'$. By the bound type variable convention, we can assume that any $\{\vtype_1, \dots, \vtype_n\}$ do not appear in $\subs$. Then $\subs(\ttype) = \forall\vtype_1::\subs(\gkind_1)\cdots\forall\vtype_n::\subs(\gkind_n).\subs(\mtype')$. Let $\subs'$ be any ground substitution such that $\dom{\subs'} = \{\vtype_1, \dots, \vtype_n\}$ and $\subs'$ respects $\{\vtype_1 :: \subs(\gkind_1), \dots, \vtype_n :: \subs(\gkind_n)\}$. Then $\subs' \circ \subs$ is a ground substitution that respects $\kenv'\cup\{\vtype_1 :: \gkind_1, \dots, \vtype_n :: \gkind_n\} = \kenv$, and $\eta$ is a $(\subs' \circ \subs)(\tenv_\vterm)$-\textit{environment}. Therefore, by the induction hypothesis for $\gterm_1$, if $\eta(\gterm_1) \downarrow \gterm_1'$, then $\gterm_1' = \val_1 \in \pred^{\subs'(\subs(\mtype'))}$ and, by the definition of $\pred$, $\val_1 \in \pred^{\subs(\ttype)}$. Now, suppose that $\eta(\letw{\vterm}{\gterm_1}{\gterm_2}) \downarrow \gterm'$. By the definition of evaluation contexts, $\eta(\gterm_1) \downarrow \gterm_1'$ and $(\letw{\vterm}{\gterm_1'}{\eta(\gterm_2)}) \downarrow \gterm'$. This means that $\gterm_1' = \val_1 \in \pred^{\subs(\ttype)}$. By the definition of evaluation contexts, $[\val_1/\vterm](\eta(\gterm_2)) \downarrow \gterm'$, i.\,e., $\eta\{\vterm \mapsto \val_1\}(\gterm_2) \downarrow \gterm'$. Since $\eta\{ \vterm \mapsto \val_1\}$ is a $\subs(\tenv \cup \{ \vterm : \ttype\})$-\textit{environment}, then $\gterm \in \pred^{\subs(\mtype)}$.
        \item Case (LetEv): Similar to the case for (Let).
        \item Case (Rec): Suppose $\kenv, \tenv \vdash \{\glab_1 = \gterm_1, \dots, \glab_n = \gterm_n\} : \{\glab_1 : \mtype_1, \dots, \glab_n : \mtype_n\}$ is derived from $\kenv, \tenv \vdash \gterm_i : \mtype_i, (1 \leq i \leq n)$. Now, also suppose that $\eta(\{\glab_1 = \gterm_1, \dots, \glab_n = \gterm_n\}) \downarrow \gterm$. By the definition of evaluation contexts, $(\{\glab_1 = \eta(\gterm_1), \dots, \glab_n = \eta(\gterm_n)\}) \downarrow \gterm'$. But, by the induction hypothesis, $\eta(\gterm_i) \downarrow \gterm_i'$, $\gterm_i' = \val_i \in \pred^{\subs(\mtype_i)}$, for some value $\val_i$. Thus, by the definition of $\pred$, we have $\gterm' \in \pred^{\{\glab_1 : \subs(\mtype_1), \dots, \glab_n : \subs(\mtype_n)\}}$.
        \item Case (Sel): Suppose $\kenv, \tenv \vdash \sel{\gterm}{\glab} : \mtype$ is derived from $\kenv, \tenv \vdash \gterm : \mtype'$ and $\kenv \vdash \mtype' :: \kind{\glab :: \mtype}$. Now, also suppose that $\eta(\sel{\gterm}{\glab}) \downarrow \gterm'$. By the definition of evaluation contexts, we have that $\eta(\gterm) \downarrow \gterm''$ and $\sel{\gterm''}{\glab} \downarrow \gterm'$. By the induction hypothesis, we have that $\gterm'' = \val \in \pred^{\subs(\tau)}$ for some value $\val$. Since $\subs$ is a ground substitution that respects $\kenv$, by Lemma~\ref{lem1}, we have that $\emptyset \vdash \subs(\mtype') :: \kind{\glab : \subs(\mtype)}$. This implies that $\subs(\mtype')$ is a ground record type of the form $\{\dots, \glab : \subs(\mtype), \dots\}$. Thus, by the definition of $\pred$, $\val = \{\dots, \glab = \val', \dots\}$, $\val' \in \pred^{\subs(\mtype)}$. But $\sel{\{\dots, \glab = \val', \dots\}}{\glab} \downarrow \val'$.
        \item Case (Modif): Suppose $\kenv, \tenv \vdash \modw{\gterm_1}{\glab}{\gterm_2} : \mtype$ is derived from $\kenv, \tenv \vdash \gterm_1 : \mtype$, $\kenv, \tenv \vdash \gterm_2 : \mtype'$, and $\kenv \Vdash \mtype :: \kind{\glab : \mtype'}$. Now, also suppose $\eta(\modw{\gterm_1}{\glab}{\gterm_2}) \downarrow \gterm'$. By the definition of evaluation contexts, $\eta(\gterm_1) \downarrow \gterm_1'$ and $\modw{\gterm_1'}{\glab}{\eta(\gterm_2})) \downarrow \gterm'$. By the induction hypothesis for $\gterm_1$, $\gterm_1' = \val_1 \in \pred^{\subs(\mtype)}$, for some value $\val_1$. Since $\subs$ is a ground substitution that respects $\kenv$, by Lemma~\ref{lem1}, $\emptyset \Vdash \subs(\mtype) :: \kind{\glab : \subs(\mtype')}$. This implies that $\subs(\mtype)$ is a ground record type of the form $\{\dots, \glab : \subs(\mtype'), \dots\}$. By the definition of $\pred$, $\val_1 = \{\dots, \glab = \val, \dots\}, \val \in \pred^{\subs(\mtype')}$. By the definition of evaluation contexts, $\eta(\gterm_2) \downarrow \gterm_2'$ and $\modw{\gterm_1'}{\glab}{\gterm_2'} \downarrow \gterm'$. By the induction hypothesis for $\gterm_2$, $\gterm_2' = \val_2 \in \pred^{\subs(\mtype')}$, for some value $\val_2$. But $\modw{\{\dots, \glab = \val, \dots\}}{\glab}{\val_2} \downarrow \{\dots, \glab : \val_2, \dots\}$ and $\{\dots, \glab = \val_2, \dots\} \in \pred^{\{\dots, \glab : \subs(\mtype'), \dots\}} = \pred^{\subs(\mtype)}$.
    \end{itemize}
    \vspace{-0.3in}
\end{proof}
From this theorem, we have the following corollary, which states that a well-typed program of type $\ttype$ evaluates to a value of type $\ttype$. In particular, this means that a well-typed program will not produce run time type errors.

\begin{corollary}
    If $\emptyset, \emptyset \vdash \gterm : \ttype$ and $\gterm \downarrow \gterm'$ then $\gterm'$ is a value of type $\ttype$.
\end{corollary}

\section{A type inference algorithm for \EVL}
\label{sec:ti}
We now adapt Ohori's $\infer{\kenv}{\tenv}{\gterm}$ inference algorithm to our language. It uses a refinement of Robinson's unification algorithm~\cite{Robinson65} that considers kind constraints on type variables. We start by discussing \emph{kinded unification} for \EVL types.

A \emph{kinded set of equations} is a pair $\kuni{\kenv}{\geqs}$, where $\kenv$ is a kinding environment and $\geqs$ is a set of pairs of types $ (\mtype_1, \mtype_2)$ that are well-formed under $\kenv$.
A \emph{kinded substitution} $\ksubs{\kenv}{\subs}$ is a unifier of a kinded set of equations $\kuni{\kenv}{\geqs}$, if every type that appears in $\geqs$ respects $\kenv$, and $\forall (\mtype_1, \mtype_2) \in \geqs, \subs(\mtype_1) = \subs(\mtype_2)$ ($\subs$ satisfies $\geqs$). A kinded substitution $\ksubs{\kenv_1}{\subs_1}$ is the \emph{most general unifier} of $\kuni{\kenv}{\geqs}$ if it is a unifier of $\kuni{\kenv}{\geqs}$ and if for any other unifier $\ksubs{\kenv_2}{\subs_2}$ of $\kuni{\kenv}{\geqs}$ there is some substitution $\subs_3$ such that $\ksubs{\kenv_2}{\subs_3}$ respects $\kenv_1$ and $\subs_2 = \subs_3 \circ \subs_1$.

The \emph{kinded unification algorithm}, $\unify{\geqs}{\kenv}$, is defined by the transformation rules in Figure~\ref{fig:kindunif}. Each rule is of the form $(\geqs_1, \kenv_1, \subs_1) \Ra (\geqs_2, \kenv_2, \subs_2)$, where $\geqs_1,\geqs_2$ are sets of pairs of types, $\kenv_1, \kenv_2$ are kinding environments, and $\subs_1,\subs_2$ are substitutions. After a transformation step, $\geqs_2$ keeps the set of pairs of types to be unified, $\kenv_2$ specifies kind constraints to be verified, and $\subs_2$ is the substitution resulting from unifying the pairs of types that have been removed from $\geqs$. Given a  kinded set of equations $(\kenv_1, \geqs_1)$ the algorithm $\unify{\geqs_1}{\kenv_1}$ proceeds by applying the transformation rules to $(\geqs_1, \kenv_1, \emptyset)$, until no more rules can be applied, resulting in a triple $(\geqs, \kenv, \subs)$. If $\geqs = \emptyset$ then it returns the pair $(\kenv, \subs)$, otherwise it reports failure.
\begin{figure}
\[{\small
  \begin{array}{rrcl}
   & (\geqs \cup \{(\mtype, \mtype)\}, \kenv, \subs)& \Ra& (\geqs, \kenv, \subs) \\ \\
   & (\geqs \cup \{(\vtype, \mtype)\}, \kenv \cup \{(\vtype, \ukind)\}, \subs) &\Ra& ([\mtype/\vtype] \geqs, [\mtype/\vtype] \kenv, [\mtype/\vtype] \subs \cup \{(\vtype, \mtype)\}) \\ \\
   & (\geqs \cup \{(\mtype, \vtype)\}, \kenv \cup \{(\vtype, \ukind)\}, \subs) & \Ra& ([\mtype/\vtype]\geqs, [\mtype/\vtype]\kenv, [\mtype/\vtype]\subs \cup \{(\vtype, \mtype)\}) \\ \\
   & (\geqs \cup \{(\vtype_1, \vtype_2)\}, \kenv \cup \{(\vtype_1, \kind{F_1}), (\vtype_2, \kind{F_2})\}, \subs) & \Ra & ([\vtype_2/\vtype_1](\geqs \cup \{(F_1(l), F_2(l)) \mid l \in \dom{F_1} \cap \dom{F_2}\}), \\
    &&&[\vtype_2/\vtype_1](\kenv) \cup \{(\vtype_2, [\vtype_2/\vtype_1](\kind{F_1 \pm F_2}))\}, \\
& & & \ [\vtype_2/\vtype_1](\subs)\cup \{(\vtype_1, \vtype_2)\}) \\ \\
   &  (\geqs \cup \{(\vtype, \{F_2\})\}, \kenv \cup \{(\vtype, \kind{F_1})\}, \subs) &\Ra& ([\{F_2\}/\vtype](\geqs \cup \{(F_1(l), F_2(l)) \mid l \in \dom{F_1}\}), \\
    &&& [\{F_2\}/\vtype](\kenv), \\
    &&& [\{F_2\}/\vtype](\subs) \cup \{(\vtype, \{F_2)\})\}) \\
    &&&\text{if} \ \dom{F_1} \subseteq \dom{F_2} \ \text{and} \ \vtype \not \in \text{FTV}(\{F_2\}) \\ \\
    & (\geqs \cup \{(\{F_2\}, \vtype)\}, \kenv \cup \{(\vtype, \kind{F_1})\}, \subs) & \Ra & ([\{F_2\}/\vtype](\geqs \cup \{(F_1(l), F_2(l)) \mid l \in \dom{F_1}\}), \\
    &&&[\{F_2\}/\vtype](\kenv), \\
    &&&[\{F_2\}/\vtype](\subs) \cup \{(\vtype, \{F_2)\})\}) \\
    &&&\text{if} \ \dom{F_1} \subseteq \dom{F_2} \ \text{and} \ \vtype \not \in \text{FTV}(\{F_2\}) \\ \\
 &(\geqs \cup \{(\{F_1\}, \{F_2\})\}, \kenv, \subs)& \Ra& (\geqs \cup \{(F_1(l), F_2(l)) \mid l \in \dom{F_1}\}, \kenv, \subs) \\
    &&&\text{if} \ \dom{F_1} = \dom{F_2} \\ \\
& (\geqs \cup \{(\atype{\mtype^1_1}{\mtype^2_1}, \atype{\mtype^1_2}{\mtype^2_2}\}, \kenv, \subs) &\Ra& (\geqs \cup \{(\mtype^1_1, \mtype^1_2), (\mtype^2_1, \mtype^2_2)\}, \kenv, \subs)
  \end{array}}
  \]
  \caption{Kinded Unification}
  \label{fig:kindunif}
\end{figure}

\begin{example} Let $\alpha_1$ and $\alpha_2$ be two type variables and $\kenv = \{\alpha_1 :: \kind{location : \alpha_3}$, $\alpha_2 :: \kind{fire\_danger : \strtype, location : \strtype}, \alpha_3 :: \ukind\}$.
{\small
\begin{align*}
    &(\{(\alpha_1, \alpha_2)\}, \{(\alpha_1, \kind{location : \alpha_3}), (\alpha_2, \kind{fire\_danger : \strtype, location : \strtype}), (\alpha_3, \ukind)\}, \{\}) \\
    &\Ra (\{(\alpha_3, \strtype)\}, \{(\alpha_2, \kind{fire\_danger : \strtype, location : \alpha_3}), (\alpha_3, \ukind)\}, \{(\alpha_1, \alpha_2)\}) \\
    &\Ra (\{\}, \{(\alpha_2, \kind{fire\_danger : \strtype, location : \strtype})\}, \{(\alpha_1, \alpha_2), (\alpha_3, \strtype)\})
\end{align*}}
The most general unifier between $\alpha_1$ and $\alpha_2$ is the kinded substitution $(\{\alpha_2 :: \kind{fire\_danger : \strtype, location : \strtype}\}, [\alpha_2/\alpha_1, \strtype/\alpha_3])$.
\end{example}
Following Ohori's notation, in the kinded  unification algorithm we use pairs in the representation of substitutions and kind assignments. 
Also, note that the unification algorithm in~\cite{Ohori95} has a kind assignment as an extra parameter, which is used to record the solved kind constraints encountered through the unification process. However, we choose to omit this parameter because its information is only used in the proofs in~\cite{Ohori95} but not in the unification process itself.

In~\cite{Ohori95}, the correctness and completeness of the kinded unification algorithm was proved, in the sense that it takes any kinded set of equations and computes its most general unifier if one exists and reports failure otherwise.

The \emph{type inference algorithm}, $\infer{\kenv}{\tenv}{\gterm}$, is defined in Figure~\ref{fig:typeinf}. Given a  kinding environment $\kenv$, a typing environment $\tenv$, and an \EVL term $\gterm$, then $\infer{\kenv_1}{\tenv}{\gterm}= (\kenv', \subs, \ttype)$, such that $\ttype$ is the type of $\gterm$ under the kinding environment $\kenv'$ and typing environment $\subs(\tenv)$. It is implicitly assumed that the inference algorithm fails if unification or any of the recursive calls on subterms fails.
 \begin{figure}[htbp]
   \[
   {\small
   \begin{array}{rcl}
     \infer{\kenv}{\tenv}{\cterm^{\btype}} &=& (\kenv,\id,\btype)\\
     \infer{\kenv}{\tenv}{\vterm} &=& \text{if} \ \vterm \not \in \dom{\tenv} \ \text{then} \ \textit{fail} \\
    &&\text{else let} \ \forall \vtype_1::\gkind_1 \cdots \forall \vtype_n::\gkind_n.\mtype = \tenv(\vterm), \\
    && \qquad \ \ \ \ \subs = [\beta_1/\vtype_1, \dots, \beta_n/\vtype_n] \ (\beta_1, \dots, \beta_n \ \text{are fresh}) \\
    && \qquad \text{in} \ (\kenv \cup \{\beta_1::\subs(\gkind_1), \dots, \beta_n::\subs(\gkind_n)\}, \id, \subs(\mtype)) \\
   \infer{\kenv}{\tenv}{\app{\gterm_1}{\gterm_2}} &=&  \text{let} \ (\kenv_1, \subs_1, \mtype_1) = \infer{\kenv}{\tenv}{\gterm_1} \\
    &&\ \ \ \ (\kenv_2, \subs_2, \mtype_2) = \infer{\kenv_1}{\subs_1(\tenv)}{\gterm_2} \\
    &&\ \ \ \ (\kenv_3, \subs_3) = \unify{\kenv_2\{\alpha :: \ukind\}}{\{(\subs_2(\mtype_1), \atype{\mtype_2}{\vtype})\}} \ (\vtype \ \text{is fresh})\\
    &&\text{in} \ (\kenv_3, \subs_3 \circ \subs_2 \circ \subs_1, \subs_3(\vtype)) \\ 
     \infer{\kenv}{\tenv}{\abs{\vterm}{\gterm}} &=& \text{let} \ (\kenv_1, \subs_1, \mtype) = \infer{\kenv \cup \{\vtype::\ukind\}}{\tenv \cup \{\vterm : \vtype\}}{\gterm} \ (\vtype \ \text{fresh}) \\
    && \text{in} \ (\kenv_1, \subs_1, \atype{\subs_1(\vtype)}{\mtype}) \\ 
       \infer{\kenv}{\tenv}{\letw{x}{\gterm_1}{\gterm_2}} &=& \text{let} \ (\kenv_1, \subs_1, \mtype_1) = \infer{\kenv}{\tenv}{\gterm_1} \\
    &&\ \ \ \ (\kenv'_1, \sigma) = \cls{\kenv_1}{\subs_1(\tenv)}{\mtype_1} \\
    &&\ \ \ \  (\kenv_2, \subs_2, \mtype_2) = \infer{\kenv'_1}{\subs_1(\tenv) \cup \{\vterm : \ttype\}}{\gterm_2} \\
    &&\text{in} \ (\kenv_2, \subs_2 \circ \subs_1, \mtype_2) \\ 
    \infer{\kenv}{\tenv}{\letEv{x}{\gterm_1}{\gterm_2}} &=&\text{let} \ (\kenv_1, \subs_1, \etype) = \infer{\kenv}{\tenv}{\gterm_1} \\
    &&\ \ \ \ (\kenv'_1, \ttype) = \cls{\kenv_1}{\subs_1(\tenv)}{\etype} \\
    &&\ \ \ \ (\kenv_2, \subs_2, \mtype_2) = \infer{\kenv'_1}{\subs_1(\tenv) \cup \{\vterm : \ttype\}}{\gterm_2} \\
    &&\text{in} \ (\kenv_2, \subs_2 \circ \subs_1, \mtype_2) \\
   \infer{\kenv}{\tenv}{\{l_1 = \gterm_1, \dots, l_n = \gterm_n\}} &=& \text{let} \ (\kenv_1, \subs_1, \mtype_1) = \infer{\kenv}{\tenv}{\gterm_1} \\
    && \ \ \ \  (\kenv_i, \subs_i, \mtype_i) = \infer{\kenv_{i-1}}{\subs_{i-1} \circ \cdots \circ \subs_1(\tenv)}{\gterm_i} \ (2 \le i \le n) \\
    &&\text{in} \ (\kenv_n, \subs_n \circ \cdots \circ \subs_2 \circ \subs_1, \\
    &&\ \ \ \ \{l_1 : \subs_n \circ \cdots \circ \subs_2(\mtype_1), \dots, l_i : \subs_n \circ \cdots \circ \subs_{i+1}(\mtype_i), \dots, l_n : \mtype_n\}) \\
     \infer{\kenv}{\tenv}{\sel{\gterm}{l}} &=& \text{let} \ (\kenv_1, \subs_1, \mtype_1) = \infer{\kenv}{\tenv}{\gterm} \\
    &&\ \ \ \ (\kenv_2, \subs_2) = \unify{\kenv_1 \cup \{\vtype_1 :: \ukind, \vtype_2 :: \kind{l : \vtype_1}\}}{\{(\vtype_2, \mtype_1)\}} \ (\vtype_1, \vtype_2 \ \text{fresh}) \\
    &&\text{in} \ (\kenv_2, \subs_2 \circ \subs_1, \subs_2(\vtype_1))\\
     \infer{\kenv}{\tenv}{\modw{\gterm_1}{l}{\gterm_2}} &=& \text{let} \ (\kenv_1, \subs_1, \mtype_1) = \infer{\kenv}{\tenv}{\gterm_1} \\
    &&\ \ \ \ (\kenv_2, \subs_2, \mtype_2) = \infer{\kenv_1}{\subs_1(\tenv)}{\gterm_2} \\
     &&\ \ \ \ (\kenv_3, \subs_3) = \unify{\kenv_2 \cup \{\vtype_1 :: \ukind, \vtype_2 :: \kind{l : \vtype_1}\}}{\{(\vtype_1, \mtype_2), (\vtype_2, \subs_2(\mtype_1))\}} \\
     &&\ \ \ \ (\vtype_1, \vtype_2 \ \text{are fresh}) \\
    &&\text{in} \ (\kenv_3, \subs_3 \circ \subs_2 \circ \subs_1, \subs_3(\vtype_2)) \\
    \infer{\kenv}{\tenv}{\cond{\gterm_1}{\gterm_2}{\gterm_3}} &=& \text{let} \ (\kenv_1, \subs_1, \mtype_1) = \infer{\kenv}{\tenv}{\gterm_1} \\
    &&\ \ \ \ (\kenv_2, \subs_2) = \unify{\kenv_1}{\{(\mtype_1, \bltype)\}} \\
    &&\ \ \ \ (\kenv_3, \subs_3, \mtype_2) = \infer{\kenv_2}{\subs_2 \circ \subs_1 (\tenv)}{\gterm_2} \\
    &&\ \ \ \ (\kenv_4, \subs_4, \mtype_3) = \infer{\kenv_3}{\subs_3 \circ \subs_2 \circ \subs_1 (\tenv)}{\gterm_3} \\
    &&\ \ \ \  (\kenv_5, \subs_5) = \unify{\kenv_4}{\{(\subs_4(\mtype_2), \mtype_3)\}} \\
    &&\text{in} \ (\kenv_5, \subs_5 \circ \subs_4 \circ \subs_3 \circ \subs_2 \circ \subs_1, \subs_5 \circ \subs_4(\mtype_2))
   \end{array}}
   \]
   \ \\
   \caption{Type inference algorithm}
   \label{fig:typeinf}
\vspace{-0.2in}
\end{figure}
 
\begin{example}
Following Example~\ref{ex:typesystem}, we consider $\shterm = \{location = l, fire\_danger = d\}$, $\shtype_1 = \{location : \alpha_1, fire\_danger : \alpha_2\}$, $\shtype_2 = \forall\alpha_1 :: \ukind.\forall\alpha_2 :: \ukind.\atype{\alpha_1}{\atype{\alpha_2}{\shtype_1}}$, $\shtype_3 = \{location : \strtype, fire\_danger : \strtype\}$, and $\shtype_4 = \atype{\strtype}{\atype{\strtype}{\shtype_3}}$, we further consider $\shtype_5 = \{location : \alpha_3, fire\_danger : \alpha_4\}$, $\shtype_6 = \forall\alpha_3::\ukind.\forall\alpha_4::\ukind.\atype{\alpha_3}{\atype{\alpha_4}{\shtype_5}}$, $\shtype_7 = \{location : \alpha_5, fire\_danger : \alpha_6\}$, $\shsubs_1 =  [\strtype/\alpha_5, \strtype/\alpha_3, \alpha_6/\alpha_4]$, and $\shsubs_2 = [\strtype/\alpha_6, \strtype/\alpha_5, \strtype/\alpha_3, \strtype/\alpha_4, \strtype/\alpha_2, \strtype/\alpha_1]$. A run of the algorithm for $\letEv{FireDanger}{\abs{l}{\abs{d}{\shterm}}}{FireDanger \ \porto^\strtype \ \low^\strtype}$
is given in Figure~\ref{fig:typerun}.

\begin{figure}[htbp]
{\small
\begin{align*}
    &\infer{\{\}}{\{\}}{\letEv{FireDanger}{\abs{l}{\abs{d}{\shterm}}}{FireDanger \ \porto^\strtype \ \low^\strtype}} = (\{\}, \shsubs_2, \shtype_3) \\
    \\
    &\quad \infer{\{\}}{\{\}}{\abs{l}{\abs{d}{\shterm}}} = (\{\alpha_3 :: \ukind, \alpha_4 :: \ukind \}, [\alpha_4/\alpha_2, \alpha_3/\alpha_1], \atype{\alpha_3}{\atype{\alpha_4}{\shtype_5}}) \\
    &\qquad \infer{\{\alpha_1 :: \ukind\}}{\{l : \alpha_1\}}{\abs{d}{\shterm}} = (\{\alpha_3 :: \ukind, \alpha_4 :: \ukind \}, [\alpha_4/\alpha_2, \alpha_3/\alpha_1], \atype{\alpha_4}{\shtype_5}) \\
    &\qquad \quad \infer{\{\alpha_1 :: \ukind, \alpha_2 :: \ukind\}}{\{l : \alpha_1, d : \alpha_2\}}{\shterm} = (\{\alpha_3 :: \ukind, \alpha_4 :: \ukind \}, [\alpha_4/\alpha_2, \alpha_3/\alpha_1], \shtype_5) \\
    &\qquad \qquad \infer{\{\alpha_1 :: \ukind, \alpha_2 :: \ukind\}}{\{l : \alpha_1, d : \alpha_2\}}{l} = (\{\alpha_2 :: \ukind, \alpha_3 :: \ukind\}, [\alpha_3/\alpha_1], \alpha_3) \\
    &\qquad \qquad \infer{\{\alpha_2 :: \ukind, \alpha_3 :: \ukind\}}{\{l : \alpha_3, d : \alpha_2\}}{d} = (\{\alpha_3 :: \ukind, \alpha_4 :: \ukind\}, [\alpha_4/\alpha_2], \alpha_4) \\
    \\
    &\quad \cls{\{\alpha_3 :: \ukind, \alpha_4 :: \ukind\}}{\{\}}{\atype{\alpha_3}{\atype{\alpha_4}{\shtype_5}}} = (\{\}, \shtype_6) \\\\
    &\quad \infer{\{\}}{\{FireDanger : \shtype_6\}}{FireDanger \ \porto^\strtype \ \low^\strtype} = (\{\}, [\strtype/\alpha_6] \circ \shsubs_1, \shtype_3) \\
    &\qquad \infer{\{\}}{\{FireDanger : \shtype_6\}}{FireDanger \ \porto^\strtype} = (\{\}, \shsubs_1, \atype{\alpha_6}{\{location : \strtype, fire\_danger : \alpha_6\}}) \\
    &\qquad\quad \infer{\{\}}{\{FireDanger : \shtype_6\}}{FireDanger} = (\{\}, [\alpha_5/\alpha_3, \alpha_6/\alpha_4], \atype{\alpha_5}\atype{\alpha_6}{\shtype_7}) \\
    &\qquad\quad\infer{\{\}}{\{FireDanger : \shtype_6\}}{\porto^\strtype} = (\{\}, id, \strtype) \\
    &\qquad\quad \unify{\{\}}{(\alpha_5, \strtype)} = (\{\}, [\strtype/\alpha_5]) \\
    &\qquad \infer{\{\}}{\{FireDanger : \shtype_6\}}{\low^\strtype} = (\{\}, id, \strtype) \\
    &\qquad \unify{\{\}}{(\alpha_6, \strtype)} = (\{\}, [\strtype/\alpha_6])
\end{align*}}
\ \\
\caption{Type inference run for $\letEv{FireDanger}{\abs{l}{\abs{d}{M}}}{FireDanger \ \porto^\strtype \ \low^\strtype}$}
\label{fig:typerun}
\vspace{-0.2in}
\end{figure}
\end{example}
\subsection{Soundness and completeness of \textit{WK}}
In this section we prove soundness and completeness of our type inference algorithm. 

\begin{therm}
  \label{thm:sound}
  If $\infer{\kenv}{\tenv}{\gterm} = (\kenv',\subs,\mtype)$ then $(\kenv',\subs)$ respects $\kenv$ and there is a derivation in our type system such that $\kenv', \subs(\tenv) \vdash \gterm : \mtype$.
\end{therm}

\begin{proof}
  The proof is by induction on the structure of $\gterm$. We only show the case for $\cond{\gterm_1}{\gterm_2}{\gterm_3}$. The case for $\cterm^\btype$ is trivial and the remaining cases are similar to the corresponding proof in~\cite{Ohori95}.
  \begin{itemize}
  \item $\gterm \equiv \cond{\gterm_1}{\gterm_2}{\gterm_3}$. Suppose that $\infer{\kenv}{\tenv}{\cond{\gterm_1}{\gterm_2}{\gterm_3}} = (\kenv', \subs, \mtype)$. Then $\infer{\kenv}{\tenv}{\cond{\gterm_1}{\gterm_2}{\gterm_3}} = (\kenv', \subs, \mtype)$, $\unify{\kenv_1}{\{(\mtype_1, \bltype)\}} = (\kenv_2, \subs_2)$, $\infer{\kenv_2}{\subs_2 \circ \subs_1(\tenv)}{\gterm_2} = (\kenv_3, \subs_3, \mtype_2)$, $\infer{\kenv_3}{\subs_3 \circ \subs_2 \circ \subs_1(\tenv)}{\gterm_3} = (\kenv_4, \subs_4, \mtype_3)$, $\unify{\kenv_4}{\{(\subs_4(\mtype_2), \mtype_3)\}} = (\kenv_5, \subs_5)$, $\kenv' = \kenv_5$, $\subs = \subs_5 \circ \subs_4 \circ \subs_3 \circ \subs_2 \circ \subs_1$ and $\mtype = \subs_5(\mtype_3)$. First, we show that $(\kenv_5, \subs_5 \circ \subs_4 \circ \subs_3 \circ \subs_2 \circ \subs_1)$ respects $\kenv$. By the correction of the unification algorithm, we know that $(\kenv_5, \subs_5)$ respects $\kenv_4$ and $(\kenv_2, \subs_2)$ respects $\kenv_1$. We also know that $\subs_5 \circ \subs_4(\mtype_2) = \subs_5(\mtype_3)$ and $\subs_2(\mtype_1) = \subs_2(\btype)$, \textit{i.e.} $\subs_2(\mtype_1) = \btype$. By the induction hypothesis, we know that $(\kenv_4, \subs_4)$ respects $\kenv_3$, $(\kenv_3, \subs_3)$ respects $\kenv_2$ and $(\kenv_1, \subs_1)$ respects $\kenv$. By applying Lemma \ref{lem1} as many times as needed, we have that $(\kenv_5, \subs_5 \circ \subs_4 \circ \subs_3 \circ \subs_2 \circ \subs_1)$ respects $\kenv$. Now, we are left to show that $\kenv_5, \subs_5 \circ \subs_4 \circ \subs_3 \circ \subs_2 \circ \subs_1(\tenv) \vdash \cond{\gterm_1}{\gterm_2}{\gterm_3} : \subs_5(\mtype_3)$. By the induction hypothesis, we know that $\kenv_1, \subs_1(\tenv) \vdash \gterm_1 : \mtype_1$, $\kenv_3, \subs_3 \circ \subs_2 \circ \subs_1(\tenv) \vdash \gterm_2 : \mtype_2$ and $\kenv_4, \subs_4 \circ \subs_3 \circ \subs_2 \circ \subs_1(\tenv) \vdash \gterm_3 : \mtype_3$. By Lemma \ref{lem2}, we have that $\kenv_5, \subs_5 \circ \subs_4 \circ \subs_3 \circ \subs_2 \circ \subs_1(\tenv) \vdash \gterm_1 : \subs_5 \circ \subs_4 \circ \subs_3 \circ \subs_2 \circ \subs_1(\mtype_1)$, \textit{i.e.} $\kenv_5, \subs_5 \circ \subs_4 \circ \subs_3 \circ \subs_2 \circ \subs_1(\tenv) \vdash \gterm_1 : \subs_5 \circ \subs_4 \circ \subs_3(\btype)$, $\kenv_5, \subs_5 \circ \subs_4 \circ \subs_3 \circ \subs_2 \circ \subs_1(\tenv) \vdash \gterm_1 : \btype$, $\kenv_5, \subs_5 \circ \subs_4 \circ \subs_3 \circ \subs_2 \circ \subs_1(\tenv) \vdash \gterm_2 : \subs_5 \circ \subs_4(\mtype_2)$, \textit{i.e.} $\kenv_5, \subs_5 \circ \subs_4 \circ \subs_3 \circ \subs_2 \circ \subs_1(\tenv) \vdash \gterm_2 : \subs_5(\mtype_3)$, and $\kenv_5, \subs_5 \circ \subs_4 \circ \subs_3 \circ \subs_2 \circ \subs_1(\tenv) \vdash \gterm_3 : \subs_5(\mtype_3)$. Finally, by (Cond), we have that $\kenv_5, \subs_5 \circ \subs_4 \circ \subs_3 \circ \subs_2 \circ \subs_1(\tenv) \vdash \cond{\gterm_1}{\gterm_2}{\gterm_3} : \subs_5(\mtype_3)$.
  \end{itemize}
  \vspace{-0.3in}
\end{proof}

\begin{therm}
\label{thm:complete}
  If $\infer{\kenv}{\tenv}{\gterm} = \textit{fail}$, then there is no $(\kenv_0, \subs_0)$ and $\mtype_0$ such that $(\kenv_0, \subs_0)$ respects $\kenv$ and $\kenv_0, \subs_0(\tenv) \vdash \gterm : \mtype_0$. \\
  If $\infer{\kenv}{\tenv}{\gterm} = (\kenv', \subs, \mtype)$, then if $\kenv_0, \subs_0(\tenv) \vdash \gterm : \mtype_0$ for some $(\kenv_0, \subs_0)$ and $\mtype_0$ such that $(\kenv_0, \subs_0)$ respects $\kenv$, then there is some $\subs'$ such that $(\kenv_0, \subs')$ respects $\kenv'$, $\mtype_0 = \subs'(\mtype)$, and $\subs_0(\tenv) = \subs' \circ \subs(\tenv)$.
\end{therm}

\begin{proof}
  The proof is by induction on the structure of $\gterm$. We only show the case for $\cond{\gterm_1}{\gterm_2}{\gterm_3}$. The case for $\cterm^\btype$ is trivial and the remaining cases are similar to the corresponding proof in~\cite{Ohori95}.
  \begin{itemize}

    \item $\gterm \equiv \cond{\gterm_1}{\gterm_2}{\gterm_3}$. Suppose that $\infer{\kenv}{\tenv}{\cond{\gterm_1}{\gterm_2}{\gterm_3}} = (\kenv, \subs, \mtype)$. Then $\infer{\kenv}{\tenv}{\gterm_1} = (\kenv_1, \subs_1, \mtype_1)$, $\unify{\kenv_1}{\{(\mtype_1, \bltype)\}} = (\kenv_2, \subs_2)$, $\infer{\kenv_2}{\subs_2 \circ \subs_1(\tenv)}{\gterm_2} = (\kenv_3, \subs_3, \mtype_2)$, $\infer{\kenv_3}{\subs_3 \circ \subs_2 \circ \subs_1(\tenv)}{\gterm_3} = (\kenv_4, \subs_4, \mtype_3)$, $\unify{\kenv_4}{\{(\subs_4(\mtype_2), \mtype_3)\}} = (\kenv_5, \subs_5)$, and $\kenv' = \kenv_5$, $\subs = \subs_5 \circ \subs_4 \circ \subs_3 \circ \subs_2 \circ \subs_1$, $\tau_0 = \subs_5 \circ \subs_4(\mtype_2)$. Now, suppose that $(\kenv_0, \subs_0)$ respects $\kenv$, and $\kenv_0, \subs_0(\tenv) \vdash \cond{\gterm_1}{\gterm_2}{\gterm_3} : \mtype_0$. Then $\kenv_0, \subs_0(\tenv) \vdash \gterm_1 : \bltype$, $\kenv_0, \subs_0(\tenv) \vdash \gterm_2 : \mtype_0$, and $\kenv_0, \subs_0(\tenv) \vdash \gterm_3 : \mtype_0$.\cmiguel{By applying the induction hypothesis to $\gterm_1$, we conclude that there exists some $S^{1}_0$ such that $(\kenv_0, \subs^{1}_0)$ respects $\kenv_1$ and $\bltype = \subs^{1}_0(\mtype_1)$ and $\subs_0(\tenv) = \subs^{1}_0(\subs_1(\tenv))$. Since $\subs^{1}_0$ is a unifier of $\bltype$ and $\mtype_1$, then, by the correctness and completeness of the unification algorithm, there is a $\subs^{2}$ such that $(\kenv_0, \subs_2)$ respects $\kenv_2$ and $\subs^{1}_0 = \subs^{2}_0 \circ \subs_2$.}{} By applying the induction hypothesis to $\gterm_2$, we conclude that there exists some $\subs^{3}_0$ such that $(\kenv_0, \subs^{3}_0)$ respects $\kenv_3$, and $\mtype_0 = \subs^{3}_0(\tenv)$ and $\subs_0(\tenv) = \subs^{3}_0 \circ \subs_3 \circ \subs_2 \circ \subs_1$. Now, by applying the induction hypothesis to $\gterm_3$, we conclude that there exists some $\subs^{4}_0$ such that $(\kenv_0, \subs^{4}_0)$ respects $\kenv_4$, and $\mtype_0 = \subs^{4}_0(\mtype_3)$ and $\subs_0(\tenv) = \subs^{4}_0 \circ \subs_4 \circ \subs_3 \circ \subs_2 \circ \subs_1(\tenv)$. It is easy to see that $\subs^{4}_0$ is a unifier of $\subs_4(\mtype_2)$ and $\mtype_3$. By the correctness and completeness of the unification algorithm, there is a $\subs^{5}_0$ such that $(\kenv_0, \subs^{5}_0)$ respects $\kenv_5$ and $\subs^{4}_0 = \subs^{5}_0 \circ \subs_5$. Then we have $\mtype_0 = \subs^{4}_0(\mtype_3) = \subs^{5}_0 \circ \subs_5 \circ \mtype_3 = \subs^{5}_0 \circ \subs_5 \circ \subs_4 \circ \mtype_2$, $\subs_0(\tenv) = \subs^{4}_0 \circ \subs_4 \circ \subs_3 \circ \subs_2 \circ \subs_1(\tenv) = \subs^{5}_0 \circ \subs_5 \circ \subs_4 \circ \subs_3 \circ \subs_2 \circ \subs_1(\tenv)$.
  \end{itemize}
  \vspace{-0.3in}
\end{proof}
\section{\EVL for event processing}
\label{sec:cep}
In this section we illustrate the use of  \EVL with examples in the context of Complex Event Processing (CEP) and specification of obligation policies.

\subsection{CEP}
The higher-order features of \EVL can be used define parameterised functions to deal with the usual CEP techniques. The canonical model~\cite{Chandy2011,Etzion10} for event processing is based on a producer-consumer model: an event processing agent (EPA) takes events from event producers and distributes them among event consumers.  \EVL is able to process raw events produced by some event processing system and generate derived events as a result. These derived events can then be passed on to an event consumer.

\subsubsection{Event processing agents}
Below we give examples to show how the standard types of event processing agents (see Section~\ref{sec:background} and \cite{Etzion10}) can be defined in \EVL.  An event processing agent is any function whose principal type is of the form  $\forall \vtype_1::\gkind_1 \cdots \forall \vtype_n::\gkind_n.\gamma$.

%
%
%
\paragraph*{Filter agents} take an incoming event object and apply a test to decide whether to discard it or whether to pass it on for processing by subsequent agents. The test is usually stateless, i.e. based solely on the content of the event instance.

\begin{example} This example defines an event processing agent that uses a higher-order filter function \emph{filter} (to be defined later) to filter events according to their location.
    \begin{lstlisting}[mathescape=true, basicstyle=\small]
    let p x = (x.location == $\porto^\strtype$) in $\lambda$x.(filter p x)
    \end{lstlisting}

\end{example}
\paragraph*{Transformation agents} can be either stateless (if events are processed without taking into account preceding or following events) or stateful (if the way  events are processed is influenced by preceding or following events). In the former case, events are processed individually. In the latter, the way events are processed can depend on preceding or succeeding events. Transformation events can be further classified as translate, split, aggregate or compose agents. 
We now give some examples of transformation agents written in \EVL.


%
%

\begin{example}This example represents a translate event processing agent that converts the temperature field of an event from Fahrenheit to  Celsius degrees.
    \begin{lstlisting}[mathescape=true, basicstyle=\small]
    farToCel x = modify(x, temperature, (x.temperature-$32.0^\floattype$)/$1.8^\floattype$)
    \end{lstlisting}
   
    Assuming the usual operational semantics for the $(\minus)$ and $(\divi)$ operators, we can evaluate the following program using the operational semantics defined for the \EVL language as follows:
    {\small
    \begin{align*}
        & \eval{\app{\abs{x}{\modif{x}{\temperature}{((\sel{x}{\temperature}) \minus 32.0^\floattype) \divi 1.8^\floattype)}}}{\{\temperature = 50.0^\floattype\}}} \\
        & \ra \eval{\modif{\{\temperature = 50.0^\floattype\}}{\temperature}{\\&\qquad\qquad\qquad\ ((\sel{\{\temperature = 50.0^\floattype\}}{\temperature}) \minus 32.0^\floattype) \divi 1.8^\floattype)}} \\
        & \ra \eval{\modif{\{\temperature = 50.0^\floattype\}}{\temperature}{((50.0^\floattype \minus 32.0^\floattype) \divi 1.8^\floattype)}} \\
        & \ra \eval{\{\temperature = 10.0^\floattype\}} \ra \{\temperature = 10.0^\floattype\}
    \end{align*}}
\end{example}

\EVL does not allow us to add or remove attributes in events. One can create new events based on attributes from incoming events as well as modify the value of existing attributes. We follow Ohori's treatment of record types, therefore we do not consider operations that extend a record with a new field or that remove an existing field from a record. 



\begin{example}
    This example defines an aggregate event processing agent that receives two events, $x$ and $y$, and outputs event $y$ with its precipitation level updated with the average of the two. 
    \begin{lstlisting}[mathescape=true, basicstyle=\small]
  avg x y = modify(y, precipitation, (x.precipitation + y.precipitation)/$2^\floattype$)
    \end{lstlisting}

\end{example}

\begin{example}
    This example defines an event processing agent that composes the partial weather information that is provided by two different sensors. One of the sensors outputs event $x$, which contains information about the temperature and wind velocity, and the other sensor outputs event $y$, which contains information about the humidity and precipitation levels. This event processing agent outputs an instance of \emph{WeatherInfo} with the complete weather information.
    \begin{lstlisting}[mathescape=true, basicstyle=\small]
     composeInfo x y = WeatherInfo x.temperature x.wind 
                                   y.humidity y.precipitation
     \end{lstlisting}

    Assuming that
    {\footnotesize
    \begin{align*}
        & \WeatherInfo : \\ 
        & \quad \atype{\floattype}{\atype{\floattype}{\atype{\floattype}{\atype{\floattype}{\{\temperature : \floattype, \wind : \floattype, \humidity : \floattype, \precipitation : \floattype\}}}}},
    \end{align*}}
    its principal typing is
    \begin{align*}
        & (\{\alpha_1 :: \kind{\temperature : \floattype, \wind : \floattype}, \alpha_2 :: \kind{\humidity : \floattype, \precipitation : \floattype}\}, \\
        & \quad \atype{\alpha_1}{\atype{\alpha_2}{\{\temperature : \floattype, \wind : \floattype, \humidity : \floattype, \precipitation : \floattype\}}}).
    \end{align*}

\end{example}

\paragraph*{Pattern Detect agents} take collections of incoming events and examine them to see if they can spot the occurrence of particular patterns.

\begin{example}The \emph{check} function in Example~\ref{ex:fire_check} is an event processing agent that generates the appropriate \emph{FireDanger} event by detecting its corresponding fire weather information.
    \begin{lstlisting}[mathescape=true, basicstyle=\small]
     check x = if (x.temperature > $29.0^\floattype$ and x.wind > $32.0^\floattype$
                   and x.humidity < $20.0^\floattype$ and x.precipitation < $50.0^\floattype$)
               then FireDanger x.location $\high^\strtype$
               else FireDanger x.location $\low^\strtype$
    \end{lstlisting} 
    Assuming,
     \[
    \begin{array}{l}
    \greater : \greaterfloattype, \\
        \lesser : \lesserfloattype, \\
        \logicand : \logicandtype,\\
        \FireDanger : \atype{\strtype}{\atype{\strtype}{\{\location : \strtype, \danger : \strtype\}}},
    \end{array}
    \]
    its principal typing is
    {\small
    \begin{align*}
        & (\{\alpha_1 :: \kind{\temperature : \floattype, \wind : \floattype, \humidity : \floattype, \precipitation : \floattype, \location : \strtype}\}, \\
        & \quad \atype{\alpha_1}{\{\location : \strtype, \danger : \strtype\}}).
    \end{align*}}

\end{example}


\subsubsection{A higher-order library for CEP}
Since \EVL is a higher-order language, we can easily define higher-order functions  to deal with a sequence of events (represented as a list of events). We now provide some of these useful higher-order functions.
\begin{itemize}
\item \texttt{filter} is a function that filters the events in the sequence according to some filtering expression:
  \begin{lstlisting}[mathescape=true, basicstyle=\small]
    filter p list = if list.empty then list
                    else if (p list.head) 
                         then (cons list.head (filter p list.tail))
                         else filter p list.tail
\end{lstlisting}

\item \texttt{transform} is a function that applies a transformation to all of the events in the sequence:
  \begin{lstlisting}[mathescape=true, basicstyle=\small]
    transform f list = if list.empty then list
                       else (cons (f list.head) (transform f list.tail))
  \end{lstlisting}

\item \texttt{aggregater} is a function that produces some output value by aggregating by right association the events of the sequence according to some binary aggregating function:
  \begin{lstlisting}[mathescape=true, basicstyle=\small]
    aggregater f z list = if list.empty then z
                          else f list.head (aggregater f z list.tail)
  \end{lstlisting}
\item \texttt{aggregatel} is very similar to \texttt{aggregater} but it aggregates the events by left association:
  \begin{lstlisting}[mathescape=true, basicstyle=\small]
    aggregatel f z list = if list.empty then z
                          else aggregatel f (f z list.head) list.tail
  \end{lstlisting}
\end{itemize}
We now give an example that illustrates several features described in this section.
\begin{example}
    Consider a sequence of events produced by sensors distributed across a number of locations. The events produced by a particular sensor contain information about the weather conditions at that sensor's location. More specifically, it contains information about the temperature (in degrees Celsius), the humidity level (as a percentage), the wind speed (in km/h) and the amount of precipitation (in mm), as well as information about its location. Now, consider an EPA that infers the fire danger of a particular location based on a given sequence of events produced by an arbitrary number of these sensors. This can be done with varying degrees of accuracy, but this is not the subject of this paper, so let us consider a simple algorithm based on the following three steps:
\begin{enumerate}
\item Filtering the events according to the specified location;
\item Aggregating the events according to the latest values of temperature, humidity and wind speed, and by the mean precipitation;
\item Producing an event that indicates whether there is fire danger at that particular location considering the values obtained in the previous step and
comparing them to their threshold levels.
\end{enumerate}

We now provide an implementation of this algorithm in \EVL:
\begin{lstlisting}[mathescape=true, basicstyle=\small]
    letEv FireDanger l d = {location = l, fire_danger = d}
    in let p x = (x.location == $\porto^\strtype$)
    in let f x y = (x.fst + $1^\inttype$, modify(y, precipitation, 
                                        (x.snd.precipitation + y.precipitation)/x.fst)) 
    in let check x = if (x.temperature > $29.0^\floattype$ and x.wind > $32.0^\floattype$ 
                         and x.humidity < $20.0^\floattype$ and x.precipitation < $50.0^\floattype$)
                     then FireDanger x.location $\high^\strtype$
                     else FireDanger x.location $\low^\strtype$
    in $\lambda$x.(check (aggregatel f ($1^\inttype$, {precipitation = 0}) (filter p x)).snd)
\end{lstlisting}
\end{example}

\subsubsection{Typing relations on events}
We now discuss the different semantic relations between events, which are captured by the \EVL typing system.

\paragraph*{Membership} A generic event $\gev_1$ is said to be a member of another generic event $\gev_2$ if the instances of $\gev_1$ are included in the instances of $\gev_2$.  In \EVL this notion is verified by the explicit subtyping relation between record types: assuming $\gev_1$ has type $\ttype_1$ and $\gev_2$ has type $\ttype_2$. $\gev_1$ is said to be a member of $\gev_2$ if, for all $\ttype_3$, if $\kenv \Vdash \ttype_1 \geq \ttype_3$ then $\kenv \Vdash \ttype_2 \geq \ttype_3$ for some $\kenv$.
\begin{example}
    Let $\gev_1$ be a generic event with type $\forall\vtype::\ukind\forall\etype::\kind{\glab_1 : \vtype}.\etype$ and $\gev_2$ be a generic event with type $\forall\etype::\kind{\glab_1 : \inttype}.\etype$. Then $\gev_2$ is a member of $\gev_1$ since for all $\ttype$,  
    \begin{align*}
        \emptyset \Vdash \forall\etype::\kind{\glab_1 : \inttype}.\etype \geq \ttype\qquad \textit{implies}\qquad\emptyset \Vdash \forall\vtype::\ukind\forall\etype::\kind{\glab_1 : \vtype}.\etype \geq \ttype.
    \end{align*}
However, $\gev_1$ is not a member of $\gev_2$ since, for $\ttype = \{\glab_1 : \floattype\}$,
    \begin{align*}
        \emptyset \Vdash \forall\vtype::\ukind\forall\etype::\kind{\glab_1 : \vtype}.\etype \geq \{\glab_1 : \floattype\}\qquad \textit{but}\qquad
        \emptyset \Vdash \forall\etype::\kind{\glab_1 : \inttype}.\etype \not\geq \{\glab_1 : \floattype\}.
    \end{align*}
\end{example}

\paragraph*{Generalization} The generalization relation indicates that an event is a generalization of another event. In the type theory of \EVL,  a (generic) event $\gev_1$ is the generalization of another even $\gev_2$ if $\kenv \Vdash \gev_1 \geq \gev_2$ for some $\kenv$.

\begin{example}
    Let $\gev_1$ be a generic event with type $\forall\vtype::\ukind\forall\etype::\kind{\glab_1 : \vtype}.\etype$ and $\gev_2$ be a generic event with type $\forall\etype::\kind{\glab_1 : \inttype}.\etype$. Then $\gev_1$ is a generalization of $\gev_2$ since
    \begin{align*}
        \emptyset \Vdash \forall\vtype::\ukind\forall\etype::\kind{\glab_1 : \vtype}.\etype \geq \forall\etype::\kind{\glab_1 : \inttype}.\etype.
    \end{align*}
\end{example}

\paragraph*{Specialization} The specialization relation indicates that an event is a specialization of another event. 
 In the type theory of \EVL this notion is the dual of the previous one: an event $\gev_1$ is a specialization of another even $\gev_2$ if $\kenv \Vdash \gev_2 \geq \gev_1$ for some $\kenv$.
 
\begin{example}
    Let $\gev_1$ be a generic event with type $\forall\vtype::\ukind\forall\etype::\kind{\glab_1 : \vtype}.\etype$ and $\gev_2$ be a generic event with type $\forall\etype::\kind{\glab_1 : \inttype}.\etype$. Then $\gev_2$ is a specialization of $\gev_1$ since
    \begin{align*}
        \emptyset \Vdash \forall\vtype::\ukind\forall\etype::\kind{\glab_1 : \vtype}.\etype \geq \forall\etype::\kind{\glab_1 : \inttype}.\etype.
    \end{align*}
\end{example}


\paragraph*{Retraction} A retraction event relationship is a property of an event referencing a second event. It indicates that the second event is a logical reversal of the event type that references it. For example, an event that starts a fire alert and the event that stops it. Unlike the previous notions, retraction is not directly addressed by \EVL.  Retraction is a notion that is also present in access control systems that deal with obligations, where the correct treatment of events is crucial. We will briefly discuss the treatment of events in obligation models in the next section.

\subsection{Event processing in obligation models}
The notion of event and an adequate processing of events is essential to the treatment of obligations in access control models. Obligations are  usually associated with some mandatory action that must be performed at a time defined by some temporal constraints or by the occurrence of an event. The Category Based Metamodel for Access Control with Obligations (CBACO)~\cite{AlvesDM14} defines an obligation as a tuple $o = (a,r,\gev_1,\gev_2)$, where $a$ is an action, $r$ a resource, and $\gev_1,\gev_2$  two generic event ($\gev_1$ triggers the obligation, and $\gev_2$ ends it). The model relies on two additional relations on events:
\begin{itemize}
\item \emph{Event Instantiation}: denoted $\ev::\gev$, meaning that the event $\ev$ is an instance of $\gev$, according to an instance relation between events and generic events.
\item \emph{Event Interval}: denoted $(\ev_1,\ev_2,h)$, meaning that the event $\ev_2$ closes the interval started by the event $\ev_1$ in an history of events $h$. 
\end{itemize}
As discussed in the previous section, the notion of event instantiation is directly captured by the \EVL type system. With respect to event intervals, this notion is closely related to the notion of retraction in CEP and was addressed in~\cite{AlvesBF15} by the definition of a closing function that describes how events are linked to subsequent events in history.  These functions are assumed to be defined for each system and are used to extract intervals from a given history. One of the motivations to develop \EVL was to provide a simple language to program such functions.

\section{Related work}
\label{sec:rw}
Alternative type systems to deal with records have been presented in the literature using row variables~\cite{Wand87}, which are variables ranging over finite sets of field types. One of the most flexible systems using row variables~\cite{Remy92a} allows for powerful operations on records, such as extending a record with a new field or removing an existing field from a record. Record extension is also available in other systems~\cite{Jategaonkar93,HarperM93,Cardelli90}, as well as record concatenation operations~\cite{HarperP91,Remy92,Wand89}, however adding these operations results in complications in the typing process. By following Ohori's approach we obtain a sound and complete efficient type system supporting the basic operations for dealing with records. Nevertheless, regarding the applicability of our language in the context of CEP, the integration of more flexible record operations in our language  is an aspect to be further investigated.

When it comes to processing flows of information, there are two main models leading the research done in this area: the data stream processing model~\cite{Babcock02} (that looks at streams of data coming from different sources to produce new data streams as output); and the complex event processing model~\cite{Luckham02} (that looks at events happening, which are then filtered and combined to produce new events). In~\cite{Cugola12}, several information processing systems were surveyed, which showed a gap between data processing languages  and event detection languages, and the need to define a minimal set of language constructors to combine both features in the same language. We believe that EVL is a good candidate to explore the gap between these two models.

Following the complex event processing model, one of the key features is the ability to derive complex events (composite) from lower-level events and several special purpose Event Query Languages (EQLs) have been proposed for that~\cite{Eckert2011}.  Complex event queries over real-time streams of RFID readings have been dealt with in~\cite{Wu06} yet again using a query language. The TESLA language~\cite{Cugola10} supports content-based event ﬁltering as well as being able to establish temporal relations on events, while providing a formal semantics based on temporal logic. 
The lack of a simple denotational semantics is a common criticism of CEP query languages~\cite{Zimmer99,Artikis17,Galton02}, with several languages not guaranteeing important language features, such as orthogonality, as well as an overlapping of definitions that make reasoning about these languages that much harder.  Recently, a formal framework based on a complex event logic (CEL) was proposed~\cite{Grez19}, with the purpose of ``giving a rigorous and efficient framework to CEP''. The authors define  well-formed  and  safe  formulas, as  syntactic  restrictions  that  characterize semantic properties, and argue that only well-formed formulas should be considered and that users should understand that all variables in a formula must be correctly defined. This notion of well-formed formulas and correctly defined variables is naturally guaranteed in a typed language like \EVL. Therefore we believe that \EVL can be used to provide formal semantics to CEP systems.

In the context of access control systems, the \emph{Obligation Specification Language} (OSL) defined in~\cite{HiltyPBSW07}, presents a language for events to  monitor and reason about data usage requirements. The paper defines the \emph{refinesEv} instance relation between events, which is based on a subset relation on labels, as is the case for the instance relation in~\cite{AlvesDM14}. The instance relation in~\cite{AlvesBF15} was defined by implicit subtyping on records but more generally using variable instantiation. In this paper we further generalise the notion of instance relation and define it formally using kinded instantiation.

Still in the context of access control, Barker et al~\cite{bsw08} have given a representation of events as  finite sets of ground 2-place facts (atoms) that describe an event, uniquely identified by $e_i , i \in \mathbb{N}$, and which includes three necessary facts: $happens(e_i,t_j)$, $act(e_i,a_l)$ and $agent(e_i,u_n)$, and $n$ non-necessary facts. This representation is claimed to be more flexible than a term-based representation with a fixed set of attributes. The language in this paper is flexible enough to encode the event representation in~\cite{bsw08}. Furthermore, the typing system allows us to guarantee any necessary facts by means of the typing information.

Our notion of events follows the approach of the \emph{Event Calculus}, where events are seen as action occurrences, or action happenings in a particular system and at particular point in time. This notion was initially introduced in~\cite{kowalski86} then latter axiomatised in~\cite{MillerS99}, and has been further used in the context of dynamic access control systems~\cite{Barker-etal:07} and in dynamic systems dealing with obligations~\cite{GelfondL08}. As in the case of~\cite{bsw08}, our flexible representation of events is capable of encoding the representation of events in the Event Calculus and the higher-order capabilities of \EVL allow us to reason about events and their effects in a particular system. 

\section{Conclusions and future work}
\label{sec:conc}
In this paper we present \EVL, a typed higher-order functional language for events with a typing system based on Ohori's record calculus, a sound and complete type inference algorithm, and a call-by-value operational semantics that preserves types.  We explore the expressiveness of our language by showing its application in the context of CEP and obligations. 

In future work, we will extend the language to consider more powerful operations on records, such as extending a record with a new field or removing an existing field from a record, which are not part of \EVL but could prove useful in both CEP and in the treatment of obligations. Ohori's main goal was to support the most common operations dealing with polymorphic records, while maintaining an efficient compilation method. Our hope is to be able to provide additional operations supporting extensibility, while maintaining efficient compilation.

Furthermore, we would like to explore extensions of \EVL with pattern matching, which is a powerful mechanism for decomposing and processing data. The ability to detect patterns is a key notion in most CEP systems, therefore adding matching primitives to \EVL would greatly improve its capability with respect to pattern detection.
\paragraph*{Competing Interests Declaration:}  The authors declare none.
  
\bibliographystyle{plain}
\bibliography{references}

\end{document}